\documentclass[11pt,letterpaper]{article}
\usepackage[letterpaper,margin=1.05in]{geometry}
\usepackage[T1]{fontenc}
\usepackage[utf8]{inputenc}
\usepackage{mathpazo}
\usepackage[scaled=0.92]{helvet}
\usepackage{courier}
\usepackage{microtype}
\usepackage{amsmath,amssymb,amsthm,mathtools,bm}
\usepackage{booktabs,array,tabularx,longtable}
\usepackage{xcolor,graphicx,tikz}
\usepackage{enumitem}
\usepackage{hyperref}
\usepackage[nameinlink,capitalise]{cleveref}
\usepackage{fancyhdr}
\usepackage{titlesec}
\usepackage{caption}

\definecolor{ink}{HTML}{17213A}
\definecolor{accent}{HTML}{A36A00}
\definecolor{muted}{HTML}{5D6472}
\definecolor{rulecolor}{HTML}{D5D9E2}
\hypersetup{colorlinks=true,linkcolor=ink,citecolor=accent,urlcolor=accent,
  pdftitle={Separated Signal Libraries: Packing, Saturation, and Joint Spectral Limits},
  pdfauthor={Marc Nunes},
  pdfsubject={Spherical codes, selection rules, and joint limits for signal ensembles}}
\titleformat{\section}{\Large\bfseries\color{ink}}{\thesection.}{0.55em}{}
\titleformat{\subsection}{\large\bfseries\color{ink}}{\thesubsection.}{0.5em}{}
\titlespacing*{\section}{0pt}{2.2ex plus .5ex}{1ex}
\titlespacing*{\subsection}{0pt}{1.6ex plus .4ex}{.7ex}
\setlist{nosep,leftmargin=1.5em}
\renewcommand{\headrulewidth}{0.35pt}
\renewcommand{\headrule}{\hbox to\headwidth{\color{rulecolor}\leaders\hrule height \headrulewidth\hfill}}

\newtheorem{theorem}{Theorem}[section]
\newtheorem{proposition}[theorem]{Proposition}
\newtheorem{corollary}[theorem]{Corollary}
\newtheorem{lemma}[theorem]{Lemma}
\theoremstyle{definition}

\newtheorem{example}[theorem]{Example}
\theoremstyle{remark}
\newtheorem{remark}[theorem]{Remark}

\newcommand{\E}{\mathbb E}
\newcommand{\R}{\mathbb R}

\newcommand{\ip}[2]{\left\langle #1,#2\right\rangle}
\newcommand{\norm}[1]{\left\lVert #1\right\rVert}
\newcommand{\one}{\mathbf 1}

\DeclareMathOperator{\Var}{Var}
\DeclareMathOperator{\tr}{tr}
\DeclareMathOperator{\spanop}{span}

\newcommand{\op}{\mathrm{op}}
\newcommand{\Sphere}{\mathbb S}
\newcommand{\Prob}{\mathbb P}
\newcommand{\HH}{\mathsf H}
\newcommand{\calM}{\mathcal M}
\newcommand{\calC}{\mathcal C}
\newcommand{\calP}{\mathsf P}

\begin{document}
\hypersetup{pageanchor=false}
\begin{titlepage}
\color{ink}
\vspace*{0.55in}
{\sffamily\bfseries\fontsize{29}{34}\selectfont Separated Signal Libraries\par}
\vspace{0.18in}
{\sffamily\fontsize{17}{22}\selectfont Packing, Saturation, and Joint Spectral Limits\par}
\vspace{0.28in}
{\color{accent}\rule{1.45in}{2pt}}\par
\vspace{0.3in}
{\large Marc Nunes}\par
\vspace{0.07in}
{\small AlphaNova\quad$\cdot$\quad\href{mailto:marc.nunes@alphanova.tech}{marc.nunes@alphanova.tech}}\par
\vspace{0.25in}
{\small Research Note\quad$\cdot$\quad Version 7\quad$\cdot$\quad September 2026}\par
\vfill
\begin{minipage}{0.95\textwidth}
\small
\textbf{Abstract.}
We study libraries of cross-sectional signals: at each date, a vector of forecasts over $d$ assets intended to predict the cross-sectional return one period ahead. A cross-sectionally demeaned, unit-normalized signal is a point on a sphere, and its history over $T$ dates is a point on a product of $T$ spheres. A pairwise correlation cap on signal histories is therefore a minimum angular separation between points on that product, and growing a library under such a cap is a packing problem. This suggests a natural question: if a large library of signals is not too pairwise correlated, and is combined by equal weighting, does the equally weighted sum (EWS) tend to a known principal-component quantity as the number of signals grows? We answer this under explicit assumptions about how the library is filled.
\par\smallskip\noindent
Separation alone guarantees nothing. It fixes no limiting distribution; a saturated library covers the sphere yet can carry a permanently biased count; and near-maximum packing on a fixed domain forces uniform volume (by the Borodachov--Hardin--Saff best-packing asymptotic), which on the unrestricted sphere gives zero mean and no distinguished leading principal component (PC1). Alignment depends on the admission rule together with the distribution of candidates. Under the uniform product-volume benchmark, screening on positive average information coefficient (IC) yields a nonzero, target-aligned mean but an isotropic second moment, whereas a positive IC margin $\beta$ makes the target the unique PC1 at every finite $T$, with an explicit three-level spectrum whose leading eigenvalue tends to $\beta^2$ while the residual levels decay as $1/T$; margins of order $T^{-1/2}$ retain a positive admission rate but a vanishing eigengap. A finite residual-spectrum criterion certifies approximate EWS--PC1 alignment for a stopped library. Gilbert--Varshamov codes show that separation permits both outcomes: exponentially large positive-IC libraries exist whose EWS is PC1, and others whose EWS is orthogonal to PC1. Finally, $\log J=o(T)$ suffices for uniform estimation of all pairwise correlations in a pool of $J$ candidates from $T$ iid dates. All derived results are proved and checked numerically; no market data are used.
\end{minipage}
\vfill
{\footnotesize\color{muted}Keywords: spherical codes; signal selection; equal weighting; principal components; packing and covering; joint limits; residual gauge; information coefficient.}
\end{titlepage}
\hypersetup{pageanchor=true}
\setcounter{page}{1}

\section{The question: what does a diversity constraint buy?}

Throughout, a signal is a cross-sectional forecast vector over $d$ assets at each date, demeaned across assets and normalized to unit length. It is intended to predict the cross-sectional return one period ahead, and its value is judged against that realized return vector, the target of \cref{sec:positiveic}. A library of 3,000 signals over 20 assets---the motivating case of \emph{Large Signal Libraries} \cite{nuneslarge}---is large compared with the 19 available neutral directions at any one date. It can nevertheless be small compared with the number of well-separated \emph{histories} those assets support. This raises a question different from the one answered by the law of large numbers: what happens when the research process deliberately avoids signals that are too similar to those already admitted?

\emph{Large Signal Libraries} separates the equally weighted sum (EWS), or its scaled average, from signal and profit-and-loss (PnL) principal components. Its population limits require a specified sampling process. The earlier paper, \emph{Signal Correlation, IC, and PnL Dependence} \cite{nunesic}, explains how alignment with the return direction (target) separates longitudinal information from transverse signal geometry. Here the information coefficient (IC) is the realized cross-sectional correlation between a normalized signal and that target. We retain both distinctions. The new object here is a library constrained by pairwise separation, and the new issue is how its selection rule changes the empirical distribution.

The principal applied case is a consistently oriented corpus: every admitted signal has positive time-averaged correlation with the target. An intended equal-weight forecasting ensemble is not an unfiltered symmetric collection containing equally many positively and negatively associated signals. We retain that unrestricted collection as a geometric benchmark, then study the positive-IC restriction explicitly in \cref{sec:positiveic}. This restriction concerns an entire signal history; it does not require the signal to be positively correlated with the target on every date.

Three interpretations of ``fill the sphere'' must be kept apart. A separated set respects the distance rule. A saturated set admits no further point without breaking the rule. A maximum set has the largest possible cardinality. A greedy admission procedure may reach saturation; that does not establish maximum cardinality. An actual research process can stop far before either saturation or maximum capacity. Its finite candidate pool and discovery order then remain part of the model. Making a set dense says where its points occur, whereas an equal-weight average depends on how many occur in each region.

The conclusions depend on the admitted domain. On the unrestricted fixed sphere, near-maximum packing has a vanishing mean. Restricting admission to positive IC removes that symmetric cancellation at fixed history length, but by itself need not create a dominant eigendirection; a positive IC margin does, under the uniform product-history benchmark, at every finite history length, and the history length then determines only how strong that direction is. Saturation alone, by contrast, can leave additional selection bias. Along the time dimension, very large positive-IC libraries exist with either agreement or disagreement between EWS and PC1, so separation decides neither; and estimating the relationships used to select them requires conditions beyond geometric feasibility.

The results are deterministic unless a probability law is stated. The packing asymptotic used in \cref{sec:smallangle} is classical \cite{bhs}; the coding argument follows the Gilbert--Varshamov method \cite{gilbert,varshamov}; and the eigenvector estimates belong to standard symmetric-operator perturbation theory \cite{daviskahan,yws}. The contribution developed here is their application to selected signal histories, the explicit contrasting constructions, and the separation of geometric, spectral, and estimation limits. \paragraph{Scope.}
This is a population and finite-configuration study, with synthetic numerical checks and no market-data results. The capacity bounds count geometrically possible histories; they do not estimate the supply of predictive signals. Target-conditioned benchmarks and coded constructions describe geometry; they are not a causal method for forecasting unknown returns. Claims about real selection rules and forecasting performance require separate temporal evaluation. Independence and uniformity are assumed only where stated; no new optimal spherical-code bound is claimed.

For orientation, $K$ is the number of admitted signals, $T$ the number of dates, $\theta$ the minimum angle between histories, and $\beta$ an average-IC floor. The geometry is defined in \cref{sec:geometry}; \cref{tab:regimes} collects the main distinctions before the proofs, and \cref{tab:notation} the recurring notation.
\begin{table}[tb]
\centering\small
\caption{What the principal regimes establish. Asset dimension is fixed throughout.}
\label{tab:regimes}
\begin{tabularx}{\textwidth}{@{}>{\raggedright\arraybackslash}p{0.30\textwidth}>{\raggedright\arraybackslash}X@{}}
\toprule
Regime & Conclusion and required qualification \\
\midrule
Fixed $T,\theta>0$ & $K$ has a finite maximum; moments must be examined at that finite size. \\
Fixed $T$, $\theta\to0$ & Near-maximum occupancy gives uniform volume on the specified admitted domain. Saturation alone does not. \\
Uniform positive half & Mean is target-aligned, but the uncentered second moment is isotropic. \\
Uniform margin $\beta>0$ & Target is unique PC1 for every finite $T$. At fixed margin, mean norm tends to $\beta$ and gap to $\beta^2$ as $T\to\infty$. \\
Margin $\beta_T=z/\sqrt{qT}$, uniform benchmark & A fixed positive $z$ gives target PC1 for every admissible finite $T$, a nonzero limiting admission rate, and vanishing mean norm and gap. \\
Fixed acute $\theta$, $T\to\infty$ & Exponential size permits either EWS--PC1 alignment or disagreement, including with positive IC. \\
$T\to\infty$, $\theta\to0$, possibly $\beta\to0$ & Specify a joint array; packing, moment, and estimation errors need uniform control. \\
\bottomrule
\end{tabularx}
\end{table}

\begin{table}[tb]
\centering\small
\caption{Recurring notation. Symbols reused with a different meaning are listed together.}
\label{tab:notation}
\begin{tabularx}{\textwidth}{@{}>{\raggedright\arraybackslash}p{0.27\textwidth}X@{}}
\toprule
Symbol & Meaning \\
\midrule
$q=d-1$, $T$, $D=qT$, $K$ & Neutral dimension, number of dates, ambient history dimension, library size \\
$\calM_{q,T}$, $n_T$ & Product of $T$ spheres carrying unit histories \eqref{eq:history}; the target history \eqref{eq:icdomain} \\
$c_{ij}$, $\vartheta_{ij}$, $\theta$ & Pairwise history correlation and angle \eqref{eq:metrics}; the minimum admitted angle \\
$c(x)=\ip{x}{n_T}$, $\bar p_i$, $\beta$ & Average IC of a history; that of signal $i$; the IC margin \\
$E_{+,T}$, $E_{0,T}$, $E_{\beta,T}$ & Positive-IC, nonnegative-IC, and margin domains \\
$\mu$, $M$, $\Sigma$ & Library mean and uncentered and centered second moments \eqref{eq:moments}; $\rho=\norm\mu^2$ \\
$\lambda_0,\lambda_L,\lambda_\perp$ & Target, longitudinal-contrast, and transverse levels \eqref{eq:marginlevels}; $\lambda_1\ge\lambda_2\ge\cdots$ are ordered eigenvalues \\
$m_{\beta,T},u_{\beta,T},A_{\beta,T}$; $\kappa$ & Conditional moments of the margin benchmark; the exponential tilt \eqref{eq:tilt} \\
$\calP_E(\theta)$, $\zeta_j$ & Maximum $\theta$-separated size in $E$; occupancy relative to it \\
$\Delta_m$, $\mathsf v_m$, $\kappa_m$; $\Delta_{\mathrm w}$ & Packing density, unit-ball volume, and the constant in \eqref{eq:bhs}; the Voronoi weight discrepancy \eqref{eq:voronoi} \\
$C$ & A code (calligraphic $\calC$) in the text; $\E_\beta p_1p_2$ inside the proof of \cref{thm:margin} \\
\bottomrule
\end{tabularx}
\end{table}

\section{Geometry and the quantities being compared}\label{sec:geometry}

Let $H=\{s\in\R^d:\one^\top s=0\}$ and $q=\dim H=d-1\ge2$. After demeaning and normalization, a signal at date $t$ is $S_{i,t}\in\Sphere(H)\simeq\Sphere^{q-1}$. All dates initially have equal weight.

\begin{samepage}
\noindent Its normalized history is
\begin{equation}\label{eq:history}
 x_i=T^{-1/2}(S_{i,1},\ldots,S_{i,T})\in\calM_{q,T}
 :=T^{-1/2}\bigl(\Sphere^{q-1}\bigr)^T\subset\Sphere^{qT-1}.
\end{equation}
\end{samepage}
The history has norm one. The product manifold has intrinsic dimension $(q-1)T$, while the surrounding vector space has dimension $D=qT$. These dimensions have different roles: the former controls small-distance packing at fixed $T$; the latter bounds matrix rank.

For unit histories define
\begin{equation}\label{eq:metrics}
 c_{ij}=\ip{x_i}{x_j}=\frac1T\sum_t\ip{S_{i,t}}{S_{j,t}},\qquad
 \vartheta_{ij}=\arccos c_{ij},\qquad
 \norm{x_i-x_j}=2\sin(\vartheta_{ij}/2).
\end{equation}
A minimum angle $0<\theta\le\pi$ is therefore equivalent to $c_{ij}\le\cos\theta$ for $i\ne j$. This is an average of datewise cosines, not an average of angles. At $60^\circ$, the upper correlation threshold is $1/2$. A rule $|c_{ij}|\le1/2$ is stronger: it also excludes almost opposite signals. Our packing results use the one-sided rule unless stated otherwise; independently changing signal signs does not generally preserve that rule or the EWS.

For a finite library $\calC=\{x_1,\ldots,x_K\}$ write
\begin{equation}\label{eq:moments}
 \mu=\frac1K\sum_i x_i,\quad
 M=\frac1K\sum_i x_i\otimes x_i,\quad
 \Sigma=\frac1K\sum_i(x_i-\mu)\otimes(x_i-\mu).
\end{equation}
Here $(u\otimes v)z=\ip{v}{z}u$. We compare the EWS direction with the top eigenvector of the \emph{uncentered} second moment $M$. Cross-sectional demeaning has already occurred; centering across signal designs would instead give $\Sigma$ and remove the mean term of interest.

If $X$ has columns $x_i$, then the time-averaged signal Gram matrix is $G=X^\top X$. The nonzero eigenvalues of $G/K$ and $M=XX^\top/K$ agree. A unit leading design eigenvector $v$ maps to the history principal component $Xv/\sqrt{K\lambda_1}$. None of these objects is the return-covariance operator. Nor is a whole-history principal component the same object as the leading eigenvector of $K^{-1}\sum_i S_{i,t}\otimes S_{i,t}$ at one date.

\section{The applied case: positive average IC}\label{sec:positiveic}

\subsection{A hemisphere of histories, not a hemisphere at every date}
Let $Q_t\in\Sphere(H)$ be the normalized target direction. Use the same orthogonal map $G_t$ for all signals at date $t$, with $G_tQ_t=e$ for a fixed unit north pole $e\in H$. Define IC directly by $p_{i,t}=\ip{S_{i,t}}{Q_t}=\ip{G_tS_{i,t}}{e}$. Expressing the histories in this common target frame gives
\begin{equation}\label{eq:icdomain}
 n_T=T^{-1/2}(e,\ldots,e),\qquad
 \bar p_i=\frac1T\sum_t p_{i,t}=\ip{x_i}{n_T}.
\end{equation}
The admissible set for a strictly positive-average-IC corpus is
\begin{equation}\label{eq:positivehistory}
 E_{+,T}=\{x\in\calM_{q,T}:\ip{x}{n_T}>0\}.
\end{equation}
It is the part of the product manifold lying in an open hemisphere of the ambient history sphere. A signal with datewise ICs $0.6$ and $-0.2$ belongs to this set. Requiring $p_{i,t}>0$ at every date would instead restrict it to a product of datewise hemispheres, a substantially smaller domain. \Cref{fig:positiveic} illustrates the distinction. The frame change itself preserves all correlations; positivity is an additional admission condition.

\begin{figure}[tb]
\centering
\includegraphics[width=0.90\textwidth]{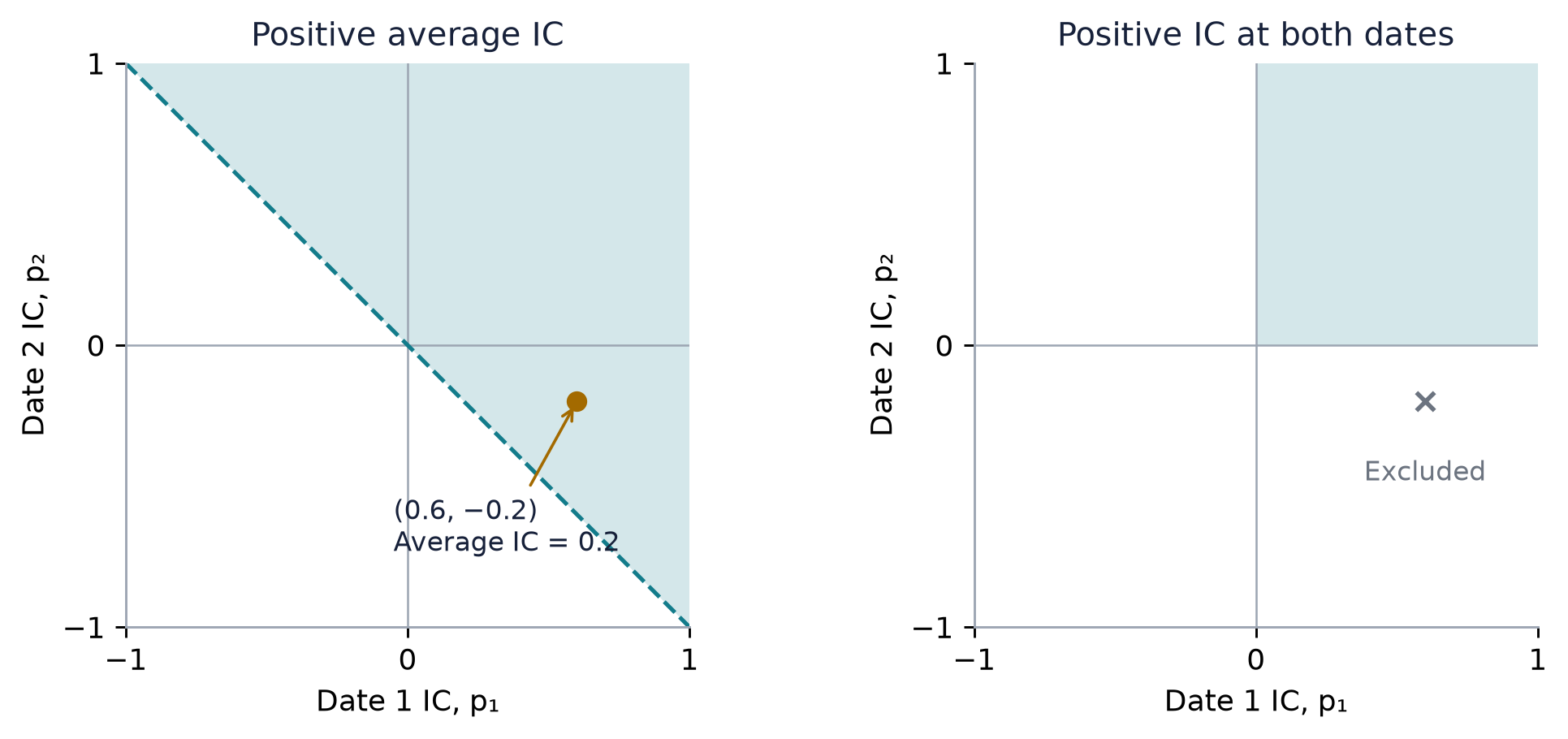}
\caption{For two dates, positive average IC permits the shaded triangle $p_1+p_2>0$, including signals with one negative date. Positivity at both dates permits only the upper-right quadrant. The marked signal has average IC $0.2$. The plot describes admissibility, not a uniform probability density over latitudes.}
\label{fig:positiveic}
\end{figure}

Negative-IC candidates can be excluded, or assigned a consistent reversed orientation before admission. The latter changes pairwise signed correlations: replacing $x_i$ by $s_ix_i$ changes $c_{ij}$ to $s_is_jc_{ij}$. A one-sided packing test must therefore be applied after orientation. This operation is separate from the common datewise target-frame transformation.

\subsection{What positivity guarantees for the ensemble}
\begin{proposition}[Positive mean IC and an IC-margin packing bound]\label{prop:positiveew}
For a nonempty finite library with $\bar p_i>0$ for every $i$, let $\bar p=K^{-1}\sum_i\bar p_i$. Then
\begin{equation}\label{eq:positiveew}
 \ip{\mu}{n_T}=\bar p>0,\qquad
 \norm\mu\ge\bar p,\qquad
 \ip{\mu/\norm\mu}{n_T}=\frac{\bar p}{\norm\mu}\ge\bar p.
\end{equation}
If its minimum angle is $\theta>0$ and $\bar p^2>\cos\theta$, then
\begin{equation}\label{eq:icpacking}
 K\le\frac{1-\cos\theta}{\bar p^2-\cos\theta}.
\end{equation}
In particular, if every $\bar p_i\ge\beta>0$, then $\norm\mu\ge\beta$ and the same bound holds with $\beta$ in place of $\bar p$ whenever $\beta^2>\cos\theta$.
\end{proposition}
\begin{proof}
Linearity gives the first identity. Cauchy--Schwarz gives $\norm\mu\ge\bar p$, and averaging unit vectors gives $\norm\mu\le1$. The pairwise bound yields
$\norm\mu^2\le[1+(K-1)\cos\theta]/K$. Combining it with $\norm\mu^2\ge\bar p^2$ and rearranging gives \eqref{eq:icpacking}.
\end{proof}

The bound is attained, at their separation angle, by vectors with a common latitude $\beta$ and regular-simplex transverse components whenever dimension permits. Their off-diagonal correlation is $\beta^2-(1-\beta^2)/(K-1)$. At $60^\circ$ separation, this bound requires corpus-average IC above $1/\sqrt2\approx0.707$ and is vacuous at the weak-IC scales motivating the paper. It becomes informative near orthogonality: $\theta=85^\circ$ and $\bar p=0.3$ give $K\le320.94$, hence $K\le320$.

Under a fixed selected population supported on $\bar p_i>0$, the population mean IC is strictly positive. A law of large numbers for that fixed population therefore gives a nonzero EWS limit. Strict positivity alone does not control arrays in which the population or history length changes: all $\bar p_i$ may remain positive while their average tends to zero. A positive lower limit for the corpus-average IC is sufficient to prevent this. The individual floor $\bar p_i\ge\beta>0$ is one convenient stronger condition, not a necessary one. A positive projection onto the target also does not force parallelism with it: a mean transverse component can remain.

These identities concern the unrenormalized EWS history and its single history-level normalization. Normalizing each date's EWS separately changes the weights given to dates and is a different statistic. Positive unweighted mean IC likewise does not by itself fix dispersion-weighted mean PnL; \cref{sec:gauge} retains that distinction.

\subsection{Selecting a positive half changes the first moment, not every spectrum}
\begin{proposition}[Antipodal selection preserves the second moment]\label{prop:antipodal}
Let $P$ be a distribution of unit histories invariant under $x\mapsto-x$, and assume $P(\ip{x}{n_T}=0)=0$. With $c(x)=\ip{x}{n_T}$, selection on $c(x)>0$ gives
\begin{equation}\label{eq:antipodal}
 \E_P[x\otimes x\mid c>0]=\E_P[x\otimes x],\qquad
 \E_P[x\mid c>0]=\E_P[\operatorname{sgn}(c)x].
\end{equation}
Its target projection is $\E_P|c|>0$. The selected mean need not be parallel to the target without additional symmetry.
\end{proposition}
\begin{proof}
The antipodal map exchanges the positive and negative events, each of probability $1/2$. It leaves $x\otimes x$ unchanged, so the two conditional second moments equal the unconditional one. The map $x\mapsto\operatorname{sgn}(c)x$ sends both halves to the positive half with its conditional distribution, giving the mean identity and its target projection.
\end{proof}

This selection uses a specified symmetric population. Folding an existing code can change its pairwise signed correlations and destroy its packing property.

\begin{theorem}[The uniform positive-IC history benchmark]\label{thm:positivehalf}
Let $P_T$ be uniform product volume on $\calM_{q,T}$ in the target frame, and retain only $c(x)>0$. Let $p_t$ be the north coordinate of a uniform point on $\Sphere^{q-1}$ at date $t$, before conditioning, so that the $p_t$ are iid with mean zero and variance $1/q$. Define
\begin{equation}\label{eq:positivehalf}
 a_{q,T}=\E\left|\frac1T\sum_t p_t\right|.
\end{equation}
The selected population has moments
\begin{equation}\label{eq:positivehalfmoments}
 \mu_+=a_{q,T}n_T\ne0,\qquad
 M_+=\frac{I_{qT}}{qT}.
\end{equation}
For fixed $q$ and $T\to\infty$,
\begin{equation}\label{eq:positivehalfrate}
 a_{q,T}\sim\sqrt{\frac{2}{\pi qT}}.
\end{equation}
For one date, $a_{q,1}=\Gamma(q/2)/[\sqrt\pi\,\Gamma((q+1)/2)]$.
\end{theorem}
\begin{proof}
Uniform product volume is antipodally symmetric and has unconditioned second moment $I_{qT}/(qT)$. The coordinate average has a continuous distribution, so \cref{prop:antipodal} preserves this second moment after selection. Independent transverse rotations force the selected mean to have only north coordinates; date permutations make those coordinates equal. Its projection onto $n_T$ is $\E|T^{-1}\sum_t p_t|$, proving \eqref{eq:positivehalfmoments}. The central limit theorem gives $Y_T:=\sqrt{qT}\,T^{-1}\sum_t p_t\Rightarrow N(0,1)$. The second moment of $Y_T$ is one for every $T$, giving uniform integrability of the absolute values; their expectations therefore tend to $\sqrt{2/\pi}$. The one-date expression follows by integrating the sphere-coordinate density proportional to $(1-p^2)^{(q-3)/2}$.
\end{proof}

The centered covariance makes the distinction sharper:
\begin{equation}\label{eq:centeredpositivehalf}
 \Sigma_+=\frac{I_{qT}}{qT}-a_{q,T}^2 n_T\otimes n_T.
\end{equation}
The target is its unique \emph{smallest} eigendirection, with eigenvalue $1/(qT)-a_{q,T}^2=\Var(c\mid c>0)$; every orthogonal direction has eigenvalue $1/(qT)$. Centering across designs reverses the role of the very direction the ensemble retains.

In original coordinates, the benchmark mean signal at date $t$ is $a_{q,T}Q_t$. That target alignment uses the benchmark symmetries, beyond the positive-IC admission rule. Thus filling the positive-IC half at fixed history length can produce a target-aligned EWS while leaving every ambient direction tied in the uncentered second moment. It does not produce a unique PC1. Growing $T$ can then make that EWS small despite strictly positive mean IC for every retained history. The benchmark conditions on realized finite-history IC; it does not assert positive long-run IC for the underlying iid date process. A uniform positive margin excludes this particular limit.

The strict set $E_{+,T}$ is open. For compact-domain packing statements, use its closure $E_{0,T}=\{c\ge0\}$ or a closed positive-margin domain $E_{\beta,T}=\{c\ge\beta\}$ with $0<\beta<1$. These are compact subsets of a smooth manifold with positive volume and satisfy the rectifiability hypotheses specified in \cref{sec:bhs}. Explicitly, a closed subset of the compact $C^1$ manifold $\calM_{q,T}$ is a finite union of Lipschitz images of compact sets, one for each of finitely many coordinate charts, so its Minkowski content equals its Hausdorff measure \cite[Theorem~3.2.39]{federer} and the best-packing theorem of Borodachov, Hardin, and Saff \cite[Theorem~2.2]{bhs} applies to it. The boundary $\{c=0\}$ has zero product volume; boundary smoothness is not required. The fixed-domain volume limit in \cref{thm:uniform}, applied to $E_{0,T}$, therefore gives the benchmark in \cref{thm:positivehalf}; the same conclusion holds for strictly positive-IC codes approaching the capacity of that closure. This claim concerns a fixed $T$ and near-maximum occupancy of the admitted domain, not every positive-IC sampling rule.

Finally, positive IC does not invalidate the disagreement example later in the paper. In \cref{thm:orthogonalcode}, choosing $m_t=Q_t$ makes every signal's IC equal to $\sqrt{0.1}>0$ at every date, yet its unique signal PC1 is orthogonal to EWS. The competing common component is transverse to the target. Its signed loadings cancel in the average but reinforce its second moment; the corpus does not contain a negatively predictive half.

\subsection{A positive margin gives a leading target under uniform sampling}\label{sec:margin}
A positive screen and a positive margin have different spectral effects. Let $\E_\beta$ denote uniform product volume conditioned on $c=T^{-1}\sum_t p_t\ge\beta$, where $0<\beta<1$. The theorem below characterizes the mean and second moment of this conditional law; for them, put
\[
 m_{\beta,T}=\E_\beta c,\qquad u_{\beta,T}=\E_\beta c^2,\qquad A_{\beta,T}=\E_\beta p_1^2.
\]
Date permutations and independent transverse rotations preserve this conditional law. These are symmetries of the uniform benchmark, not consequences of admission alone. Within the product geometry their group is $O(q-1)^T\rtimes S_T$; the ambient sphere has a larger stabilizer that does not preserve the product manifold.

\begin{theorem}[A positive margin makes the target uniquely leading]\label{thm:margin}
For every $q\ge2$, $T\ge1$, and $0<\beta<1$, uniform product volume conditioned on $c\ge\beta$ has mean $m_{\beta,T}n_T$ and unique leading second-moment direction $n_T$. For $T\ge2$, its spectral levels are (the case $T=1$ is stated after the proof):
\begin{equation}\label{eq:marginlevels}
\begin{aligned}
 \lambda_0&=u_{\beta,T} &&\text{on }\spanop\{n_T\},\\
 \lambda_L&=\frac{A_{\beta,T}-u_{\beta,T}}{T-1}
 &&\text{on longitudinal contrasts }(\dim=T-1),\\
 \lambda_\perp&=\frac{1-A_{\beta,T}}{T(q-1)}
 &&\text{on transverse histories }(\dim=T(q-1)).
\end{aligned}
\end{equation}
Here longitudinal contrasts have the form $T^{-1/2}(b_1e,\ldots,b_Te)$ with $\sum_t b_t=0$. The strict inequalities $\lambda_0>\lambda_L$ and $\lambda_0>\lambda_\perp$ hold for every positive margin. Writing $\lambda_1\ge\lambda_2\ge\cdots$ for the ordered eigenvalues, this means $\lambda_1=\lambda_0$ and $\lambda_2=\max(\lambda_L,\lambda_\perp)$. A coarse quantitative lower bound, useful only when its numerator is positive, is
\begin{equation}\label{eq:marginfinitegap}
 T\beta^2>1
 \quad\Longrightarrow\quad
 \lambda_1-\lambda_2\ge\frac{T\beta^2-1}{T-1}>0.
\end{equation}
\end{theorem}
\begin{proof}
Write $C=\E_\beta p_1p_2$. The longitudinal block has diagonal $A_{\beta,T}/T$ and off-diagonal $C/T$, while
$u_{\beta,T}=[A_{\beta,T}+(T-1)C]/T$. Its constant and contrast eigenvectors give the first two levels. Transverse rotations remove cross terms and make each datewise transverse block scalar; its trace is $(1-A_{\beta,T})/T$. This gives the third level. The same rotations and date permutations give the stated mean. Since $\beta^2\le u_{\beta,T}\le A_{\beta,T}\le1$,
\[
 \lambda_L\le\frac{1-\beta^2}{T-1},\qquad
 \lambda_\perp\le\frac{1-\beta^2}{T(q-1)}\le\frac{1-\beta^2}{T-1}.
\]
Subtracting from $\lambda_0\ge\beta^2$ proves \eqref{eq:marginfinitegap}. Strict dominance for all positive margins, including when this bound is uninformative, is proved by the exact gap identities in \cref{app:finitegap}. The argument also covers $T=1$, when only the transverse comparison is present.
\end{proof}

For $T=1$, the contrast space is absent: $\lambda_0=A_{\beta,1}=\ell_\beta$ and $\lambda_\perp=(1-\ell_\beta)/(q-1)$, the cap formula \eqref{eq:capmoments}. The theorem closes the finite-history question: a positive margin selects the leading target direction without any minimum date count. The bound \eqref{eq:marginfinitegap} controls the size of a gap; it does not locate its onset.

The limiting constants can be computed with one-dimensional integrals. Let $p$ be one coordinate of a uniform point on $\Sphere^{q-1}$ and define
\begin{equation}\label{eq:tilt}
 Z_q(\kappa)=\E e^{\kappa p},\quad
 \frac{d}{d\kappa}\log Z_q(\kappa)=\beta,\quad
 dF_\kappa(p)=\frac{e^{\kappa p}}{Z_q(\kappa)}dF(p).
\end{equation}
There is a unique solution $\kappa=\kappa(q,\beta)>0$, because the tilted mean increases continuously from zero to one. Put $A_\kappa=\E_\kappa p^2$ and $v_\kappa=A_\kappa-\beta^2>0$.

\begin{corollary}[A fixed margin survives a growing history]\label{cor:marginlimit}
For fixed $q\ge2$ and $0<\beta<1$, as $T\to\infty$,
\begin{equation}\label{eq:marginlimit}
\begin{aligned}
 m_{\beta,T}&\longrightarrow\beta,&
 \lambda_0&\longrightarrow\beta^2,\\
 T\lambda_L&\longrightarrow v_\kappa,&
 T\lambda_\perp&\longrightarrow\frac{1-A_\kappa}{q-1}=\frac{\beta}{\kappa}.
\end{aligned}
\end{equation}
Consequently the EWS norm tends to $\beta$, its direction is the target history, and the leading gap tends to $\beta^2$. The two residual coefficients generally differ, although they agree to leading order as $\beta\downarrow0$; \cref{app:refinements} gives the expansion.
\end{corollary}
\begin{proof}
The exponential-tilting argument in \cref{app:tilting} proves $c\to\beta$ in conditional probability and $A_{\beta,T}\to A_\kappa$. Boundedness gives $m_{\beta,T}\to\beta$ and $u_{\beta,T}\to\beta^2$, and \eqref{eq:marginlevels} gives the spectral limits. Integrating the derivative of $e^{\kappa p}(1-p^2)^{(q-1)/2}$ over $[-1,1]$ gives $\kappa(1-A_\kappa)=(q-1)\beta$.
\end{proof}

The finite-history correction can also be made precise. For fixed $q$ and $0<\beta<1$,
\begin{equation}\label{eq:finitecorrection}
 m_{\beta,T}=\beta+\frac{1}{\kappa T}+o(T^{-1}),\qquad
 \lambda_0=\beta^2+\frac{2\beta}{\kappa T}+o(T^{-1}).
\end{equation}
Under the retained law, $T(c-\beta)$ tends to an exponential variable of rate $\kappa$, with convergence of its first two moments. \Cref{app:refinements} proves this correction and the associated admission prefactor. It is a fixed-margin result, not a uniform approximation as $\beta$ shrinks.

Concretely, the margin-screened population has a one-date IC marginal that tends to the exponentially tilted law $F_\kappa$, whose mean is $\beta$. This is an asymptotic description of a retained signal at one date, not an assertion that its conditioned dates are independent at finite $T$.

For each fixed $T$, near-maximum small-angle packing in $E_{\beta,T}$ has these conditional moments in the packing limit by \cref{thm:uniform}. Taking $T\to\infty$ \emph{after} that limit gives \cref{cor:marginlimit}. A joint sequence needs its own error control: if its mean and second moment differ from the conditional benchmark by $o(\beta)$ and $o(\beta^2)$ respectively, fixed positive $\beta$ and \cref{cor:marginlimit} give the same limiting direction and gap. We make no uniform packing claim when $\beta$ shrinks.

\begin{figure}[tb]
\centering
\includegraphics[width=0.96\textwidth]{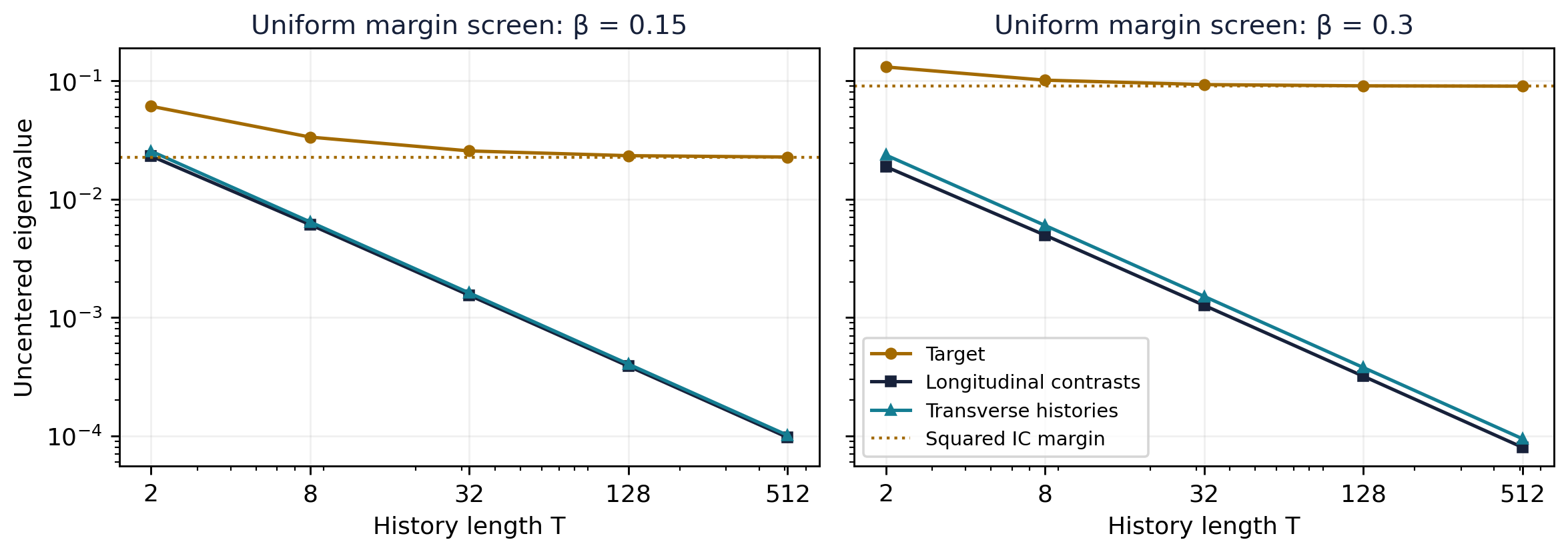}
\caption{The three margin-screened spectral levels for $q=19$, estimated by importance sampling from the conditional benchmark (40,000 proposals per point). The target level approaches $\beta^2$ while both residual levels decay. \Cref{thm:margin} proves target dominance for every $T$ and positive margin; the coarse bound \eqref{eq:marginfinitegap} becomes informative at $T\ge45$ for $\beta=0.15$. Equation~\eqref{eq:finitecorrection} explains the approach to $\beta^2$. The source bundle reports effective sample sizes and separate exact quadrature checks.}
\label{fig:margin}
\end{figure}

The uniform candidate population pays for a fixed margin through rarity: \cref{app:tilting} gives admission probability $\exp\{-T I(\beta)+o(T)\}$ with $I(\beta)>0$. Thus the fixed-margin event becomes exponentially rare under uniform candidates.

The result is useful precisely because it states what must be modeled: a persistent IC margin \emph{and} the distribution of residual structure. An arbitrary margin-screened corpus need not have the uniform conditional law. The orthogonal-PC1 construction in \cref{thm:orthogonalcode} has a fixed positive margin and still disagrees.

\subsection{Small margins: spectral strength and admission rate}\label{sec:smallmargin}
Every positive margin already gives a leading target by \cref{thm:margin}. The remaining questions are how strong that direction is and how many candidates survive the screen. The fixed-margin residual coefficients both approach $1/q$ as $\beta\downarrow0$, suggesting the scale $qT\beta^2$. \Cref{app:refinements} records their expansions. The following joint limit treats a shrinking margin directly.

\begin{proposition}[A margin at the central-limit scale]\label{prop:localscreen}
Fix $q\ge2$ and $z\ge0$, put $\beta_T=z/\sqrt{qT}$, and use uniform product volume conditioned on $c\ge\beta_T$. Let $\varphi$ be the standard normal density, $\overline\Phi(z)=\int_z^\infty\varphi(y)\,dy$, and $h(z)=\varphi(z)/\overline\Phi(z)$. As $T\to\infty$,
\begin{equation}\label{eq:localscreen}
\begin{aligned}
 \Prob(c\ge\beta_T)&\longrightarrow\overline\Phi(z),&
 \sqrt{qT}\,m_{\beta_T,T}&\longrightarrow h(z),\\
 qT\lambda_0&\longrightarrow1+zh(z),&
 qT\lambda_L, qT\lambda_\perp&\longrightarrow1.
\end{aligned}
\end{equation}
For every fixed $z>0$, the target is the unique PC1 for every $T$ with $qT>z^2$, that is, whenever $0<\beta_T<1$, and $qT(\lambda_1-\lambda_2)\to zh(z)>0$. Its absolute gap and mean norm still vanish. At $z=0$, the positive-half second moment is exactly isotropic for every $T$.
\end{proposition}
\begin{proof}
The central limit theorem applied to $Y_T=\sqrt{qT}\,c$ gives the admission limit. Uniform fourth-moment bounds justify convergence of its first two conditional moments, while removing one date shows $A_{\beta_T,T}\to1/q$. The three-level formulas then give \eqref{eq:localscreen}; \cref{app:localscreen} supplies the details.
\end{proof}

At $z=0$, $h(0)=\sqrt{2/\pi}$ recovers the mean scale of \cref{thm:positivehalf}; the second moment is exactly isotropic at every $T$. For fixed $z>0$, both residual levels are $1/(qT)+o(T^{-1})$. This leading-order equality does not identify which is $\lambda_2$. The fixed-margin expansions in \cref{app:refinements} do not supply that higher-order comparison in this joint limit; use the exact maximum in \eqref{eq:marginlevels}.

Thus $qT\beta_T^2=z^2$ controls relative spectral strength and selectivity, rather than the onset of PC1. At $z=1$, the limiting admission rate is $15.87\%$ and the target-to-residual eigenvalue ratio is $2.525$. A unique leading target does not require discarding almost all candidates. A \emph{persistent absolute gap} under a fixed positive margin is the more selective regime described by \cref{cor:marginlimit}.

Admission probabilities need a tail prefactor as well as an exponential rate. \Cref{app:refinements} gives the formula and the $q=19$, $\beta=0.1$, $T=100$ example; \cref{tab:weakmargin} keeps the weak-IC benchmark here.

\begin{table}[tb]
\centering\small
\caption{Uniform margin screen over 20 assets: $q=19$, $\beta=0.03$. The first row uses cap quadrature; the remaining rows use refined tilted convolution. At $T=1$ the residual comparison uses $\lambda_\perp$ alone.}
\label{tab:weakmargin}
\begin{tabular}{rrrr}
\toprule
$T$ & Admission probability & $\lambda_0/\max(\lambda_L,\lambda_\perp)$ & Gap \\
\midrule
1 & 0.4500 & 1.117 & $0.00613$ \\
2 & 0.4280 & 1.171 & $0.00449$ \\
10 & 0.3401 & 1.447 & $0.00235$ \\
60 & 0.1556 & 2.557 & $0.00136$ \\
100 & 0.0955 & 3.326 & $0.00122$ \\
1,000 & $1.77\times10^{-5}$ & 19.022 & $0.000948$ \\
\bottomrule
\end{tabular}

\end{table}
\Cref{tab:weakmargin} shows how conservative \eqref{eq:marginfinitegap} can be. Its lower bound becomes positive only at $T\ge1{,}112$, whereas \cref{thm:margin} proves a unique leading target at every $T$. At $T=60$, near the scale $qT\beta^2\approx1$, the target level is about $2.56$ times the largest residual level and roughly $15.6\%$ of the uniform pool passes.

\section{Finite packing and finite stopping}\label{sec:finite}

For nonempty compact $E\subset\Sphere^{D-1}$ let $\calP_E(\theta)$ be the maximum size of a subset with pairwise angles at least $\theta$. A code is \emph{saturated in $E$} if no point of $E$ can be added. Saturation is always relative to the admissible domain: exhausting one generator's candidates does not imply saturation in all of $\calM_{q,T}$.

\begin{proposition}[Finite capacity and coverage]\label{prop:coverage}
For fixed nonempty compact $E$ and $0<\theta\le\pi$, $\calP_E(\theta)<\infty$ and a maximum code exists. Every saturated code has covering radius less than $\theta$ in angular distance. On the full sphere,
\begin{equation}\label{eq:volumebounds}
 \frac1{v_D(\theta)}\le\calP_{\Sphere^{D-1}}(\theta)
 \le\frac1{v_D(\theta/2)},\qquad
 v_D(r)=\frac{\int_0^r\sin^{D-2}u\,du}{\int_0^\pi\sin^{D-2}u\,du}.
\end{equation}
\end{proposition}
\begin{proof}
A finite cover of $E$ by balls of diameter smaller than $\theta$ bounds the number of separated points. For each feasible cardinality, the closed separation constraints define a compact subset of a finite Cartesian power of $E$; hence the largest feasible cardinality is attained. If a point were at distance at least $\theta$ from every code point, it could be added. Compactness then makes the covering radius strictly smaller than $\theta$. For the upper bound, caps of radius $\theta/2$ have disjoint interiors. For the lower bound, the caps of radius $\theta$ around a saturated code cover the sphere.
\end{proof}

For the 20-asset case, $\Sphere^{18}$ is a sphere in 19 dimensions, and $\calP_{\Sphere^{18}}(\pi/3)$ is the \emph{kissing number in dimension 19}. Ho's construction \cite[Theorem 1, v2]{ho} gives
\[
 \calP_{\Sphere^{18}}(\pi/3)\ge11{,}948.
\]
Leijenhorst and de Laat report the numerical semidefinite upper-bound value $24{,}417.472$, hence the integer value $24{,}417$ \cite[Table 1]{delaat}. Their paper states that interval-arithmetic certification was not performed, so we retain this as a reported numerical upper bound. Cohn and Li's updated table also records $11{,}948$ and $24{,}417$ \cite[Table 1.1, v3]{cohnli}, and Cohn's online table of kissing-number bounds \cite{cohntable} aggregates the current values in one place. The arXiv versions and published source were checked on 11 September 2026. These figures concern one-sided correlations at a single date; the exact kissing number remains undetermined.

When $K$ stops at a finite maximum, there is no $K\to\infty$ theorem to invoke. One can still compute the exact mean, spectrum, row sums, and residual bounds below. Repeating draws from that fixed code converges to its chosen weighting, but does not enlarge the separated set. Even a maximum configuration may have multiple equally good realizations, and its finite moments must be determined rather than inferred from cardinality alone.

\section{What shrinking the angle does on a fixed sphere}\label{sec:smallangle}

\subsection{Separation alone does not select a distribution}
Let $\theta_j\downarrow0$ and $K_j\to\infty$. Permitting points to lie closer together does not require them to visit the whole sphere. They can remain in a cap, along a curve, or in shrinking neighborhoods of selected directions. More generally, any probability measure on a sphere can be approximated by equal-weight measures on distinct finite point sets: approximate its masses by rational weights, then replace repeated atoms by nearby distinct points. Choosing a positive separation threshold below each set's minimum distance produces a separated sequence with that prescribed weak limit. Thus no limiting distribution follows from $K_j\to\infty$ and $\theta_j\to0$ alone.

\subsection{Saturation gives coverage, but equal weights can remain biased}
By \cref{prop:coverage}, saturated codes with $\theta_j\to0$ approach the sphere in Hausdorff distance. This establishes the geometric sense in which their support fills it. It does not establish uniform counting density.

\begin{example}[A saturated nonuniform sequence]\label{ex:saturated}
On the unit circle let $\theta_n=\pi/(3n)$. Place $3n$ points at angles $j\pi/(3n)$, $0\le j<3n$, and $2n$ points at angles $\pi+j\pi/(2n)$, $0\le j<2n$. All consecutive gaps are either $\theta_n$ or $3\theta_n/2$. The code is separated, and no gap is large enough to insert a point at distance $\theta_n$ from both neighbors. It is therefore saturated.

Its empirical measure tends to the mixture assigning mass $3/5$ uniformly to the upper semicircle and mass $2/5$ uniformly to the lower one. Direct integration gives
\begin{equation}\label{eq:circlelimit}
 \mu_n\longrightarrow\left(0,\frac{2}{5\pi}\right),\qquad
 M_n=\frac12 I_2\quad\text{for every }n.
\end{equation}
The code fills the circle, retains a nonzero mean, and has a degenerate limiting leading eigenspace. It contains $5n$ points, whereas the maximum at this threshold is $6n$.
\end{example}

\begin{figure}[tb]
\centering
\includegraphics[width=0.94\textwidth]{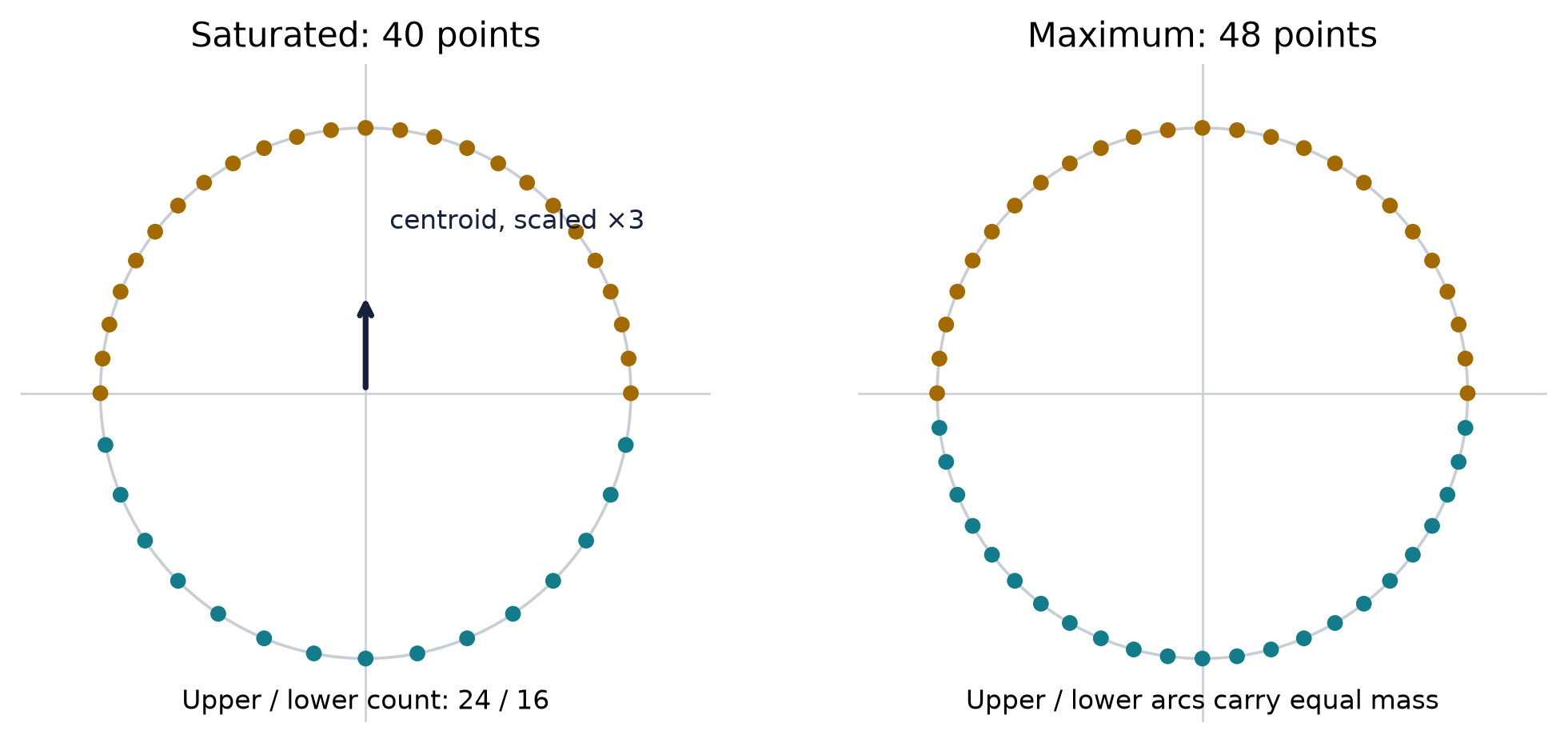}
\caption{Both configurations cover the circle increasingly closely as $n$ grows. The saturated construction keeps a $3:2$ count imbalance between semicircles. The maximum construction has equal gaps and a zero centroid. The arrow on the left marks the finite centroid; it is magnified only in length for visibility.}
\label{fig:saturation}
\end{figure}

The distinction can also be measured through Voronoi cells. Let $V_i$ be the cell assigned to point $x_i$ under nearest-neighbor allocation, with any measurable tie convention, and let $\sigma$ be uniform surface probability. Set $w_i=\sigma(V_i)$ and $\Delta_{\mathrm w}=\sum_i|K^{-1}-w_i|$. If the maximum chordal distance within the allocation is $r$, then
\begin{equation}\label{eq:voronoi}
 \norm{\mu}\le r+\Delta_{\mathrm w},\qquad
 \norm{M-I_D/D}_{\op}\le2r+\Delta_{\mathrm w}.
\end{equation}
Indeed, transporting every cell to its site changes the first moment by at most $r$ and the second by at most $2r$, using $\norm{x\otimes x-y\otimes y}_{\op}\le2\norm{x-y}$. Replacing the cell masses by equal weights costs at most $\Delta_{\mathrm w}$. Small covering radius alone controls only the first part of each error.

\subsection{Near-maximum cardinality does force uniformity}\label{sec:bhs}
We use compact subsets $E$ of a fixed compact $C^1$ embedded $m$-dimensional manifold, with $m\ge1$ and positive $m$-dimensional Hausdorff measure $\operatorname{vol}_m(E)$. Finite coordinate charts make these sets $m$-rectifiable in the Lipschitz-image sense. Compact rectifiable sets have Minkowski content equal to Hausdorff measure \cite[Theorem~3.2.39]{federer}, so they satisfy the precise hypothesis of Borodachov, Hardin, and Saff \cite[Eq.~(1.9) and Theorem 2.2]{bhs}. Inverting their best-packing-distance limit gives
\begin{equation}\label{eq:bhs}
 \calP_E^{\rm chord}(\varepsilon)
 =\kappa_m\,\operatorname{vol}_m(E)\varepsilon^{-m}(1+o(1)),
 \qquad \varepsilon\downarrow0,
\end{equation}
where $\kappa_m=2^m\Delta_m/\mathsf v_m$, $\Delta_m$ is optimal Euclidean ball-packing density, and $\mathsf v_m$ is unit-ball volume. A \emph{regular patch} below means a relatively open $B\subset E$ with $\operatorname{vol}_m(\partial_E B)=0$. No smoothness of that boundary is required; the closures of both patches are compact rectifiable sets.

\begin{theorem}[Near-maximum codes on a fixed domain]\label{thm:uniform}
Fix such an $E\subset\Sphere^{D-1}$, of positive volume and positive dimension. If $\theta_j\downarrow0$, codes $\calC_j\subset E$ are $\theta_j$-separated, and
\[
 \frac{|\calC_j|}{\calP_E(\theta_j)}\longrightarrow1,
\]
then their empirical probability measures converge weakly to normalized volume on $E$.
\end{theorem}
\begin{proof}
Put $\varepsilon_j=2\sin(\theta_j/2)$. For a regular patch $B\subset E$, the number of selected points in $B$ is at most $\calP_{\overline B}^{\rm chord}(\varepsilon_j)$. Divide by $|\calC_j|$ and use \eqref{eq:bhs} to obtain an upper limit of at most $\operatorname{vol}(B)/\operatorname{vol}(E)$. Apply the capacity bound to $E\setminus B$, whose volume is the complementary volume because $\operatorname{vol}_m(\partial_E B)=0$, to get the matching lower limit. A zero-volume patch has $\calP^{\rm chord}(\varepsilon)\varepsilon^m\to0$ by the same best-packing theorem. Partitions of arbitrarily small diameter with null boundaries can be formed from metric balls whose radii avoid the countably many positive-volume level sets. They approximate continuous functions on compact $E$, proving weak convergence.
\end{proof}

On the full fixed sphere this gives the direct answer:
\begin{equation}\label{eq:uniformsphere}
 \mu_j\to0,\qquad M_j\to I_D/D,\qquad
 \lambda_1(M_j)-\lambda_2(M_j)\to0.
\end{equation}
Thus the strongest version of filling the \emph{unrestricted} sphere produces neither a nonzero limiting EWS nor a unique population PC1. The intended positive-IC corpus instead uses the restricted domain in \cref{sec:positiveic}. The direction of a small finite centroid is not controlled by this result. Nor does a vanishing eigengap prevent a particular selected sequence from having convergent eigenvectors; it removes their unique identification from the limiting operator.

At fixed $q,T$, the same theorem applies to $\calM_{q,T}$. Its uniform volume measure is the product of the uniform datewise sphere measures. Consequently its mean is zero and its ambient second moment is $I_{qT}/(qT)$. In particular,
\[
 \calP_{\calM_{q,T}}(\theta)\asymp_{q,T}\theta^{-(q-1)T}
 \quad(\theta\downarrow0,\ q,T\text{ fixed}).
\]
The constants and the onset of this approximation can depend on $T$. It is not a uniform theorem for a growing number of dates.

\subsection{Restricting the admissible region changes the answer}
Consider a fixed spherical cap $E_b=\{x\in\Sphere^{q-1}:m^\top x\ge b\}$, with $0\le b<1$. Near-maximum small-angle codes in that cap converge to uniform cap measure by \cref{thm:uniform}. Write $a=m^\top x$. Rotational symmetry gives
\begin{equation}\label{eq:capmoments}
 \mu=\bar a_b m,\quad
 M=\ell_b\,m\otimes m+\frac{1-\ell_b}{q-1}(I-m\otimes m),\quad
 \ell_b=\frac{\int_b^1 a^2(1-a^2)^{(q-3)/2}\,da}
 {\int_b^1(1-a^2)^{(q-3)/2}\,da}.
\end{equation}
Here $\bar a_b>0$. For a hemisphere, $b=0$ and $\ell_0=1/q$: the mean is nonzero but the second moment remains isotropic. For $b>0$, conditioning on a larger absolute latitude strictly raises its squared value, so $\ell_b>1/q$ and the mean is the unique PC1 direction. As $b\downarrow0$, the gap tends to zero. If cap size and packing angle change together, moment errors must be controlled relative to that gap before claiming stable alignment.

The history-margin analogue is \cref{thm:margin,cor:marginlimit}, where the product constraint creates three spectral levels. This is how a target or a favored factor can enter through the admissible domain. For $T=1$, target-positive admission gives this hemisphere. For $T>1$, positive average IC uses the product-manifold history half in \cref{sec:positiveic}, whose measure is not uniform volume on the entire ambient hemisphere. Positive average IC alone does not require every datewise signal to lie in such a cap.

\subsection{Random arrivals and greedy admission}\label{sec:rsa}
Greedy acceptance from iid candidate arrivals, rejecting every point within the exclusion distance of an accepted point, is random sequential adsorption (RSA). Its output is not a generic saturated configuration: the arrival law and admission order induce a particular distribution over configurations.

\begin{proposition}[What rotation-invariant admission implies]\label{prop:rsa}
Start from an empty full sphere with iid uniform arrivals and a rotation-equivariant distance rule. For any nonempty accepted configuration obtained after a fixed number of arrivals, or at the infinite-input jammed limit, its normalized empirical measure $\widehat\nu_\theta$ satisfies
\[
 \E\widehat\nu_\theta=\sigma,\qquad
 \E\mu_\theta=0,\qquad \E M_\theta=I_D/D,
\]
where $\sigma$ is uniform sphere probability. If, additionally, $\Var(\int f\,d\widehat\nu_\theta)\to0$ as $\theta\to0$ for a countable uniformly dense family of continuous test functions, then $\widehat\nu_\theta\Rightarrow\sigma$ in probability, and its mean and second moment converge in probability to $0$ and $I_D/D$.
\end{proposition}
\begin{proof}
Rotating every arrival rotates the accepted configuration, without changing its law or cardinality. Its expected normalized measure is therefore the unique rotation-invariant probability $\sigma$. Under the variance assumption, Chebyshev's inequality gives convergence for each test function; finite approximation and compactness give weak convergence in probability. Coordinates and coordinate products yield the moment limits.
\end{proof}

Symmetry alone does not supply the variance assumption. For example, uniformly rotate each configuration in \cref{ex:saturated}. Its expected empirical measure is uniform, but its centroid norm still tends to $2/(5\pi)$. This is a counterexample to a symmetry-only argument, not a claim about iid-arrival RSA.

The required concentration is an established subject in RSA. Schreiber, Penrose, and Yukich \cite[Theorems 1.1--1.2]{rsa} prove mean and covariance asymptotics for saturated packing measures from homogeneous arrivals in expanding Euclidean cubes. After normalizing by the accepted count, these imply uniform spatial convergence in probability. Their proof controls spatial dependence through stabilization. In the one-dimensional parking model, the limiting occupied fraction is approximately $0.7476$, strictly below the optimal fraction one \cite[Section 1]{rsa}. Uniform limiting density therefore need not require near-maximum packing.

These results partly settle the original greedy-admission question in the homogeneous Euclidean setting. Applying them to curved product domains requires the corresponding stabilization and boundary estimates; rotational invariance by itself is not that extension. The remaining research question concerns how the candidate law, IC-screened domain, stopping rule, and growing history dimension determine the selected measure. Nonuniform arrivals are one possible source of asymmetry; the screen itself already breaks full rotational symmetry. We do not claim that every nonuniform arrival law must produce a nonuniform jammed limit.

\section{Finite criteria for mean--PC alignment}\label{sec:spectral}

All statements in this section apply to a finite stopped library. No limit, packing optimality, or independence assumption is required. The identities
\begin{equation}\label{eq:decomposition}
 M=\mu\otimes\mu+\Sigma,\qquad
 \tr M=1,\qquad\tr\Sigma=1-\norm{\mu}^2
\end{equation}
show why a count of separated signals is not enough. Residual energy can be spread thinly over many directions or remain concentrated in one component.

The packing constraint supplies a bound on mean energy. If $K>1$ and $\bar c$ is the average off-diagonal correlation, then
\begin{equation}\label{eq:meanbound}
 \rho:=\norm{\mu}^2=\frac1K+\frac{K-1}{K}\bar c
 \le\frac1K+\frac{K-1}{K}\cos\theta.
\end{equation}
It gives no positive lower bound on $\rho$. Under a two-sided rule $|c_{ij}|\le c$, one also has $\tr(M^2)\le c^2+(1-c^2)/K$. The one-sided rule does not yield that squared-correlation bound, since large negative inner products are allowed.

\begin{theorem}[Persistent mean versus largest residual component]\label{thm:alignment}
Let $\rho=\norm{\mu}^2>0$ and $\eta=\norm{\Sigma}_{\op}$. If $\eta<\rho$, then $M$ has a simple leading eigenvalue, with gap at least $\rho-\eta$, and
\begin{equation}\label{eq:tangent}
 \tan\angle(\mu,v_1(M))\le\frac{\eta}{\rho-\eta}.
\end{equation}
Angles are between unoriented axes. Therefore, for any array of finite libraries, even in changing dimensions, $\eta_j/\rho_j\to0$ implies mean--PC1 angular convergence.
\end{theorem}
\begin{proof}
Set $m=\mu/\norm{\mu}$ and split the space into $\spanop(m)\oplus m^\perp$. The blocks of $M$ are $\chi=\rho+\ip{m}{\Sigma m}\ge\rho$, $r_m=(I-m\otimes m)\Sigma m$, and $B=\Sigma|_{m^\perp}$ compressed to $m^\perp$. Thus $\norm{r_m}\le\eta$ and $\lambda_{\max}(B)\le\eta$. By the min--max principle, $\lambda_1(M)\ge\rho$ and $\lambda_2(M)\le\eta$. For a leading eigenvector $v=am+w$, $w\perp m$, the block equation gives $(\lambda_1I-B)w=a r_m$. The coefficient $a$ cannot vanish, since otherwise $\lambda_1\le\eta<\rho$. Taking the inverse and dividing by $|a|$ proves \eqref{eq:tangent}.
\end{proof}

The condition is sufficient, not necessary. Exact axial symmetry can give alignment even when $\eta/\rho$ does not tend to zero. If $\Sigma\ne0$, define its effective rank by $r_{\rm eff}=\tr\Sigma/\norm\Sigma_{\op}$. Then
\[
 \frac{\eta}{\rho}=\frac{1-\rho}{\rho\,r_{\rm eff}}.
\]
Increasing $K$ matters here only if it changes the mean or increases effective residual rank. At fixed ambient dimension, $r_{\rm eff}\le\min(K-1,D)$, so it cannot grow indefinitely.

A complementary exact calculation applies when $x_i=\sqrt\rho\,m+\sqrt{1-\rho}\,u_i$, with $u_i\perp m$, $\norm{u_i}=1$, and $\sum_i u_i=0$. Then the mean is $\sqrt\rho\,m$, and it is PC1 precisely when
\[
 \rho>(1-\rho)\lambda_{\max}\left(\frac1K\sum_i u_i\otimes u_i\right).
\]
For isotropic residuals on an $r$-dimensional subspace, the threshold is $\rho>1/(r+1)$. If instead only $|\ip{u_i}{u_j}|\le\epsilon$ is known, the residual eigenvalue is at most $[1+(K-1)\epsilon]/K$ by the Gram row-sum norm bound. These are conditions on \emph{residual} geometry, stronger than a cap on original signal correlations.

\section{The exponential capacity of histories}\label{sec:codes}

A simple construction respects every date's unit-norm constraint. Choose an orthonormal basis of $H$ and represent each coordinate of each date by a sign divided by $\sqrt q$. A binary word of length $D=qT$ gives a history in \eqref{eq:history}. Words at Hamming distance $d_{\rm H}$ produce history correlation $1-2d_{\rm H}/D$. For $0<\theta<\pi/2$ put $d_0=\lceil D(1-\cos\theta)/2\rceil$. A greedy binary-code construction gives
\begin{equation}\label{eq:gvfinite}
 \calP_{\calM_{q,T}}(\theta)
 \ge\frac{2^D}{\sum_{j=0}^{d_0-1}\binom Dj}.
\end{equation}
To prove this, choose a word and delete its Hamming ball of radius $d_0-1$, repeating until no words remain. Each selected word removes at most the denominator's number of candidates. This is the usual Gilbert argument \cite{gilbert}.

Writing $\HH_2(u)=-u\log_2u-(1-u)\log_2(1-u)$, the leading exponential rate is at least $1-\HH_2((1-\cos\theta)/2)$. For $60^\circ$ it is $0.1887218755$ per binary coordinate. \Cref{tab:gvbound} evaluates \eqref{eq:gvfinite} directly for the running example. A history that simply repeats one cross-section has only the original cross-section's packing capacity.
\begin{table}[tb]
\centering\small
\caption{The Gilbert bound \eqref{eq:gvfinite} at $\theta=60^\circ$ for the 20-asset case.}
\label{tab:gvbound}
\begin{tabular}{rrrr}
\toprule
$q$ & $T$ & $D=qT$ & $\log_{10}$ of the lower bound \\
\midrule
19 & 100 & 1,900 & 109.9182 \\
19 & 1,000 & 19,000 & 1081.8841 \\
\bottomrule
\end{tabular}
\end{table}

\begin{samepage}
For the spectral constructions below we use a random-linear-code argument of Varshamov type \cite{varshamov}, adding exact first- and second-moment balance to the distance condition.

\begin{lemma}[Separated binary words with exact two-coordinate balance]\label{lem:balancedcode}
Fix $0<\delta<1/2$ and $0<R<1-\HH_2(\delta)$. For all sufficiently large $N$, there is a binary linear code of length $N$, dimension $k=\lfloor RN\rfloor$, and minimum Hamming distance at least $\lceil\delta N\rceil$, such that its uniform sign representation $z\in\{-1,1\}^N$ satisfies
\begin{equation}\label{eq:balancedcode}
 \E z=0,\qquad\E zz^\top=I_N.
\end{equation}
\end{lemma}
\end{samepage}
\begin{proof}
Choose a $k\times N$ binary generator matrix with independent fair entries. For each nonzero message, the resulting word is uniform on $\{0,1\}^N$. A union bound shows that the chance of any nonzero message having weight below $\lceil\delta N\rceil$ is at most $2^{k-N}\sum_{j<\lceil\delta N\rceil}\binom Nj$, which tends to zero by the entropy bound and the strict rate inequality. The chance that a generator column is zero or two columns coincide is at most $[N+\binom N2]2^{-k}$, also tending to zero. Therefore a matrix avoiding both failures exists. The first condition gives injectivity and the required distance. Nonzero, distinct columns define nontrivial, pairwise independent binary linear functionals of a uniform message, giving \eqref{eq:balancedcode}.
\end{proof}

Two-coordinate balance is a deterministic property of the whole code under equal weighting. It is not a claim that arbitrarily many accepted signals are independent draws.

\section{Two large separated libraries with opposite spectral outcomes}\label{sec:constructions}

Fix orthonormal frames in $H$ at each date. In the following constructions, the common pole can be the return target or another specified direction.

\subsection{A weak common direction can become history PC1}\label{sec:positive}
Let $m_t\in\Sphere(H)$ and let $e_{1,t},\ldots,e_{q-1,t}$ be an orthonormal basis of $m_t^\perp$. Put $N=(q-1)T$, choose $0<\rho<\cos\theta$ with $\theta<\pi/2$, and use the sign words of \cref{lem:balancedcode} with
\[
 \delta=\frac{1-\cos\theta}{2(1-\rho)}<\frac12.
\]
For each word $z$ define
\begin{equation}\label{eq:positivecode}
 S_{z,t}=\sqrt\rho\,m_t+
 \sqrt{\frac{1-\rho}{q-1}}\sum_{a=1}^{q-1}z_{t,a}e_{a,t}.
\end{equation}
Every signal has unit norm. For two words at Hamming distance $d_{\rm H}$, their history correlation is $1-2(1-\rho)d_{\rm H}/N$, so the code has minimum angle at least $\theta$. The number of signals is $K=2^{\lfloor RN\rfloor}$ for any fixed admissible rate and sufficiently large $N$.

\begin{theorem}[Exact aligned codes and two different thresholds]\label{thm:alignedcode}
For \eqref{eq:positivecode}, write $m^{(T)}$ for the normalized history of the poles $m_t$. Let $P_U$ project onto the $N$-dimensional transverse history space. Then
\begin{equation}\label{eq:positivecov}
 \mu=\sqrt\rho\,m^{(T)},\qquad
 M=\rho\,m^{(T)}\otimes m^{(T)}+\frac{1-\rho}{N}P_U.
\end{equation}
The EWS is the unique history PC1 direction exactly when $\rho>1/(N+1)$. At a fixed date it is the unique signal-cloud PC1 direction exactly when $\rho>1/q$.
\end{theorem}
\begin{proof}
The sign means and cross-coordinate moments in \eqref{eq:balancedcode} remove all mixed terms and give \eqref{eq:positivecov}. Its positive eigenvalues are $\rho$ and $(1-\rho)/N$, the latter with multiplicity $N$. There are also $T-1$ zero longitudinal eigenvalues. At one date the second moment is $\rho m_t\otimes m_t+(1-\rho)(I-m_t\otimes m_t)/(q-1)$. Comparing the respective eigenvalues proves both conditions.
\end{proof}

At the level of moments this is the degenerate three-level template $A=C=u=\rho$, so $\lambda_L=0$; it realizes the same block structure as \cref{thm:margin} without being uniform conditional volume. For $q=19$ and $\rho=0.01$, the datewise threshold fails, while the history threshold holds for every $T\ge6$ whenever the required balanced code is available. Thus a history principal component can reproduce the EWS without the corresponding datewise claim being true. The exact construction uses a fixed positive $\rho$; if $\rho$ changes with $T$, code existence must be checked along that array rather than inferred from the fixed-parameter lemma.

\begin{figure}[tb]
\centering
\includegraphics[width=0.94\textwidth]{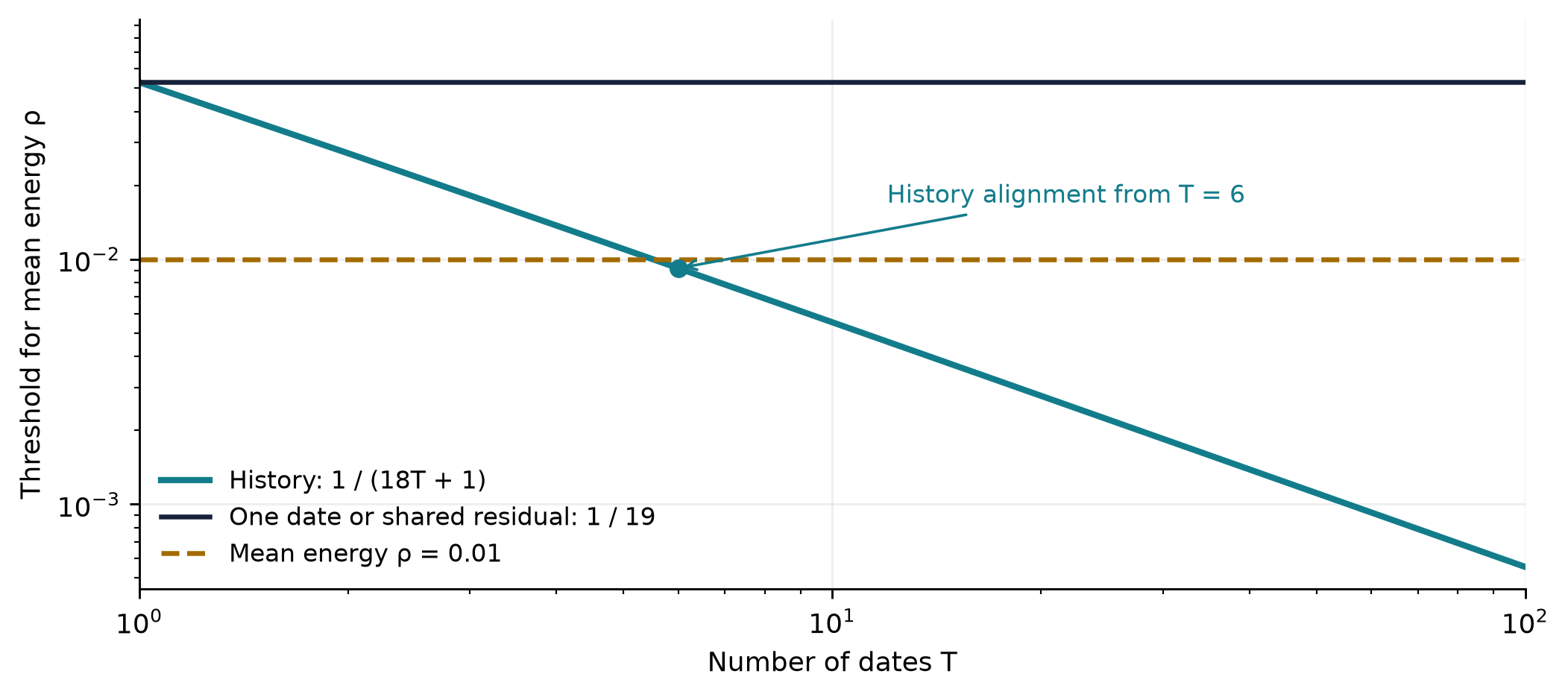}
\caption{Exact mean--PC thresholds in the balanced construction for $q=19$. Spreading residual energy over $(q-1)T$ history directions lowers the history threshold. Repeating a common residual through time does not. The marked $\rho=0.01$ is below the datewise threshold but above the history threshold from $T=6$ onward.}
\label{fig:thresholds}
\end{figure}

\subsection{Exponential size and separation can coexist with orthogonality}
Now assume $q\ge3$. Choose orthonormal $m_t,v_t$ and $q-2$ further basis vectors at each date, and set $N=(q-2)T$. Construct balanced residual histories $u_z$ from signs on those $q-2$ directions using \cref{lem:balancedcode} with $\delta=5/12$. Their pairwise correlations are at most $1/6$. Include both signs $s\in\{-1,1\}$ of a second common component:
\begin{equation}\label{eq:negativecode}
 S_{z,s,t}=\sqrt{0.1}\,m_t+\sqrt{0.3}\,s\,v_t+\sqrt{0.6}\,U_{z,t},
 \qquad \norm{U_{z,t}}=1.
\end{equation}
There are $K=2^{1+\lfloor RN\rfloor}$ signals for any $0<R<1-\HH_2(5/12)$ and sufficiently large $N$.

\begin{theorem}[Separated codes with EWS orthogonal to PC1]\label{thm:orthogonalcode}
All distinct histories in \eqref{eq:negativecode} have correlation at most $1/2$. With $m^{(T)},v^{(T)}$ denoting normalized histories,
\begin{equation}\label{eq:negativecov}
 \mu=\sqrt{0.1}\,m^{(T)},\qquad
 M=0.1\,m^{(T)}\otimes m^{(T)}+
 0.3\,v^{(T)}\otimes v^{(T)}+\frac{0.6}{N}P_U.
\end{equation}
For $N>2$, the unique history PC1 is $v^{(T)}$, orthogonal to EWS.
\end{theorem}
\begin{proof}
For different residual words the largest correlation is $0.1+0.3+0.6(1/6)=0.5$. For the same word and opposite $s$ it is $0.1-0.3+0.6=0.4$. Unit norm follows from pointwise orthogonality. Averaging both signs of $s$ and the balanced code gives \eqref{eq:negativecov}. Its largest eigenvalue is $0.3$ when $N>2$.
\end{proof}

The code lemma permits every asymptotic rate $R<1-\HH_2(5/12)=0.0201312433\ldots$ bits per residual coordinate. This remains exponential. It is much smaller than the rate $0.1887218755$ of the unrestricted $60^\circ$ construction in \cref{sec:codes}, but the two rates are per coordinate of different spaces, the $N=(q-2)T$ residual coordinates here against all $qT$ coordinates there, and are not directly comparable. The residual factor survives in the second moment even though its signs cancel in the EWS. That cancellation is exact because the corpus is symmetric by construction: every signal $(z,s)$ has a partner $(z,-s)$ with the opposite loading on $v_t$. The example shows what separation permits; it is not a claim that a generic positive-IC library cancels its transverse structure in this way. When $m_t=Q_t$, all signals have positive IC; the cancellation is in transverse factor loadings, not in predictive orientation. This example satisfies the one-sided packing rule; no two-sided correlation claim is made. Neither this code nor the aligned code is asserted to have maximum cardinality in the full history domain. That qualification matters: \cref{thm:uniform} describes near-maximum small-angle codes on a \emph{fixed} domain, not these growing-dimensional constructions.

\begin{remark}[Disagreement is stable under small perturbations]\label{rem:orthostability}
Let $\mu_0=\sqrt{0.1}\,m^{(T)}$ and $M_0$ be the moments in \cref{thm:orthogonalcode}. Its leading gap is $g_0=0.3-\max(0.1,0.6/N)>0$ for $N>2$, and equals $0.2$ for $N\ge6$. Suppose a perturbed library has moments $\widetilde\mu,\widetilde M$ with
\[
 \epsilon=\norm{\widetilde M-M_0}_{\op}<g_0/2,\qquad
 d_\mu=\norm{\widetilde\mu-\mu_0}<\sqrt{0.1}.
\]
Its leading eigendirection $\widetilde v_1$ is unique, and
\begin{equation}\label{eq:orthostability}
 \left|\ip{\widetilde\mu/\norm{\widetilde\mu}}{\widetilde v_1}\right|
 \le\frac{2\epsilon}{g_0}+\frac{2d_\mu}{\sqrt{0.1}}.
\end{equation}
Indeed, Davis--Kahan \cite{daviskahan,yws} bounds the sine of the leading-axis change by $2\epsilon/g_0$, normalization changes the mean direction by at most $2d_\mu/\sqrt{0.1}$, and the original axes are orthogonal. Thus small temporal or coordinate perturbations leave the two directions close to orthogonal when both moment errors are small. A second-moment bound alone does not control the EWS direction.

This estimate is dimension independent and applies to each finite member of a growing-dimensional array satisfying its hypotheses. It does not certify pairwise separation after perturbation. Taking the code distance fraction strictly above $5/12$ leaves strict correlation slack at a smaller positive exponential rate; that slack can absorb sufficiently small perturbations of the individual histories.
\end{remark}

\section{The target frame and the correlations it preserves}\label{sec:gauge}

Write demeaned realized returns as $a_tQ_t$, with $Q_t\in\Sphere(H)$. For dates with zero magnitude choose $Q_t$ arbitrarily, since payoff is then zero. Fix a unit pole $e\in H$ and choose a common orthogonal map $G_t$ with $G_tQ_t=e$. The same map acts on every signal at that date:
\begin{equation}\label{eq:gaugesplit}
 G_tS_{i,t}=p_{i,t}e+U_{i,t},\qquad
 p_{i,t}=\ip{S_{i,t}}{Q_t},\quad U_{i,t}\perp e.
\end{equation}
For finite dates there is no continuity issue; in a population time model a measurable frame is sufficient. The block-diagonal history map is orthogonal. It preserves all history correlations, packing constraints, spectra, and EWS--PC angles. Choosing the frame cannot create a statistical symmetry.

The residual gauge freedom---the remaining freedom to rotate transverse coordinates---is $\prod_{t=1}^T O(q-1)$. Invariance of a \emph{joint distribution over signal designs} under independent actions of this group is substantially stronger than rotational symmetry of each date's marginal distribution. If only marginal symmetry is assumed, residual components may still remain strongly related through time.

\begin{proposition}[Joint residual symmetry separates the history spectrum]\label{prop:gaugespectrum}
Suppose a distribution of unit signal histories in the target frame is invariant under independent transverse orthogonal transformations at every date. Write $\gamma_t=\E p_t$, $\gamma=(\gamma_1,\ldots,\gamma_T)^\top$, $\xi_t=\E(1-p_t^2)$, and $B=T^{-1}\E[pp^\top]$. In orthonormal longitudinal history coordinates and the transverse coordinates,
\begin{equation}\label{eq:gaugespectrum}
 \mu=T^{-1/2}\gamma,\qquad
 M=B\ \oplus\ \bigoplus_{t=1}^T\frac{\xi_t}{T(q-1)}I_{q-1}.
\end{equation}
If $\gamma\ne0$, its direction is the unique history PC1 precisely when it is the unique leading eigendirection of $B$ and
\[
 \lambda_1(B)>\max_t\frac{\xi_t}{T(q-1)}.
\]
\end{proposition}
\begin{proof}
Applying a transverse reflection at one date forces $\E U_t=0$, all $\E[p_sU_t]$ to vanish, and $\E[U_tU_s^\top]=0$ for $s\ne t$. Rotations at that date force its transverse second moment to be scalar, with trace $\xi_t$. The history normalization supplies the factor $1/T$. The remaining block is $B$, and the eigenvalue comparison is immediate.
\end{proof}

In \cref{thm:alignedcode}, constant latitude gives $B=\rho\one\one^\top/T$. Its finite code has the moments in \eqref{eq:gaugespectrum}, although the discrete code itself need not have full orthogonal invariance. The symmetry is one sufficient explanation of these moments, not a necessary one.

For contrast, in a common transverse coordinate frame let $p_t=\sqrt\rho$ and $U_t=\sqrt{1-\rho}\,Z$ at every date, where $Z$ is uniform on $\Sphere^{q-2}$. Every marginal is rotationally symmetric. However, the transverse history covariance has nonzero eigenvalues $(1-\rho)/(q-1)$, independent of $T$. The history threshold remains $1/q$. Increasing the number of repeated dates has not spread the residual variation over new directions. A different choice of gauge may change the appearance of cross-time blocks, but it cannot change these eigenvalues.

The signal, IC, and PnL correlations retain different information \cite{nuneslarge,nunesic}. For a finite library,
\begin{equation}\label{eq:payoff}
 G_{ij}=\frac1T\sum_t\bigl(p_{i,t}p_{j,t}+\ip{U_{i,t}}{U_{j,t}}\bigr),\qquad
 C^{\Pi}_{ij}=\operatorname{Cov}_T(a_tp_{i,t},a_tp_{j,t}).
\end{equation}
Here $\operatorname{Cov}_T$ uses divisor $T$. The first matrix is the signal Gram matrix. Uncentered IC similarity keeps only its longitudinal term. PnL covariance additionally weights by return magnitude and centers through time. IC Pearson correlation also centers and rescales; PnL Pearson correlation further divides by each payoff standard deviation and is undefined for zero-variance histories.

These differences are concrete in \cref{thm:orthogonalcode}. If $m_t=Q_t$, every signal has $p_{i,t}=\sqrt{0.1}$, so every payoff history is $\sqrt{0.1}\,a_t$. When $a_t$ varies, PnL covariance is rank one with equal leading design weights, while the leading signal history $v^{(T)}$ has zero payoff at every date. Signal separation has not produced PnL separation.

For completeness, let $L_T:\R^{qT}\to\R^T$ map a history vector $x$ to $(a_t\ip{x_t}{Q_t})_t$, and let $P_0=I_T-\one\one^\top/T$. Its output on \eqref{eq:history} is the raw PnL vector divided by $\sqrt T$. Set $\mathcal L_T=P_0L_T$ and $\psi_i=\mathcal L_Tx_i$. The centered PnL second moment across designs and its equal-weight history are
\[
 A_\Pi:=\frac1K\sum_i\psi_i\otimes\psi_i=\mathcal L_T M\mathcal L_T^*,\qquad
 \overline\psi:=\frac1K\sum_i\psi_i=\mathcal L_T\mu.
\]
All these finite-dimensional maps are bounded. Along growing histories their norms can change with $\max_t|a_t|$; no uniform boundedness or infinite-dimensional adjoint identity is assumed. Finally, the asset-space return second moment $T^{-1}\sum_t a_t^2Q_t\otimes Q_t$ (or its temporally centered covariance) remains a separate operator. Target alignment is not return-PC alignment.

\section{Joint limits without silently changing the problem}\label{sec:joint}

\subsection{Specify an array, not an informal limit}
Let $j$ index finite problems with parameters $(K_j,T_j,\theta_j)$ and code $\calC_j\subset\calM_{q,T_j}$. For the applied corpus, additionally specify the positive-IC domain and any margin $\beta_j$. Keep $q$ fixed unless another dimension limit is explicitly introduced. The regimes are summarized in \cref{tab:regimes}.

For the unrestricted fixed domain define occupancy $\zeta_j=K_j/\calP_{\calM_{q,T}}(\theta_j)$. This is a theoretical comparison to maximum capacity, often unknown computationally. For a positive-IC corpus, replace that denominator by capacity in its specified compact admitted domain. A greedy stopping rule does not certify $\zeta_j\approx1$. When $T$ varies, even $\zeta_j\to1$ cannot be inserted into \cref{thm:uniform} without control uniform in $T$: that theorem fixes the domain before taking its limit.

Angles and spectral ratios can be compared within each member of the array without identifying vector spaces across $j$. To claim convergence to a particular history vector requires a common function space or specified embeddings and a consistent temporal interpretation. Bounded norm alone does not give strong compactness in an infinite-dimensional Hilbert space.

There is also a sampling obstruction. In a separable metric space, an iid infinite sequence from any fixed probability law almost surely fails a fixed positive separation rule. A countable cover by balls of diameter less than the separation contains a ball of positive probability; infinitely many draws enter it. Thus an indefinitely separated library is a selected, dependent configuration, not an unrestricted iid sample from one fixed law.

\subsection{Mean decay and residual dimension compete}
The balanced model makes the rate issue explicit. Put $N_j=(q-1)T_j$ and allow $\rho_j>0$. Whenever the stated moments are realized,
\begin{equation}\label{eq:criticalrate}
 \text{mean is unique history PC1}\quad\Longleftrightarrow\quad
 (N_j+1)\rho_j>1.
\end{equation}
Thus $\rho_j=N_j^{-a}$ eventually gives a leading mean for $a<1$ and a leading transverse space for $a>1$. The exactly tied sequence is $\rho_j=1/(N_j+1)$. The mean norm tends to zero in all these examples when $a>0$; every signal still has positive IC $\sqrt{\rho_j}$, but the IC margin shrinks. Angular alignment does not establish a nonvanishing raw ensemble.

The aligned-code gap is $g_j=\rho_j-(1-\rho_j)/N_j$. If a model approximation perturbs its second moment by $e_j$ in operator norm, a sufficient condition for persistence of its leading direction is $e_j/g_j\to0$, with $g_j>0$, by the Davis--Kahan bound \cite{daviskahan,yws}. The packing parameters alone do not control either quantity.

Taking $\theta\to0$ at each fixed $T$ under near-maximum packing gives the uniform product moments $\mu=0$, $M=I_{qT}/(qT)$. Growing $T$ after that does not recreate a mean. Conversely, the fixed-positive-angle coded constructions can retain a mean as $T$ grows. These are different selections and limit paths. We do not assert a formal noncommutation theorem for an unspecified common family. Establishing such a theorem would require defining that family and both limits first.

\subsection{Geometric capacity and statistical resolution}\label{sec:estimation}
To isolate temporal estimation, suppose a pool of $J\ge2$ candidate signal processes is fixed before an independent evaluation sample of $T$ iid dates. All candidates have unit norm at each date, but candidates may be arbitrarily dependent on one another. Let $c_{ij}$ be population time-averaged similarity and $\widehat c_{ij}$ its sample estimate. Hoeffding's inequality \cite{hoeffding} and a union bound give
\begin{equation}\label{eq:uniformerror}
 \Prob\!\left(\max_{i<j}|\widehat c_{ij}-c_{ij}|>\epsilon\right)
 \le J(J-1)e^{-T\epsilon^2/2}.
\end{equation}
Thus, with probability at least $1-\alpha$, all errors are bounded by
\[
 \epsilon_{J,T}(\alpha)=\sqrt{\frac{2}{T}\log\frac{J(J-1)}{\alpha}}.
\]
The maximum is over the candidate pool, not just the $K$ signals ultimately retained. Because the event is simultaneous, a correlation-based selection from that fixed pool may depend on the evaluation matrix. If models themselves are tuned on those dates, this calculation does not apply without an additional argument or fresh evaluation data.

On the simultaneous event, enforcing $\widehat c_{ij}\le\cos\theta-\epsilon_{J,T}$ certifies $c_{ij}\le\cos\theta$ for every retained pair. For any selected subset, normalize its population and empirical Gram matrices by $K$. Their operator-norm difference is at most $\epsilon_{J,T}$, by the maximum absolute row sum. If the population normalized Gram has leading gap $g>0$, its leading design vector obeys
\begin{equation}\label{eq:gramerror}
 \sin\angle(\widehat v_1,v_1)\le\frac{2\epsilon_{J,T}}{g}.
\end{equation}
This is a design-space estimate; synthesis into histories and any datewise claim use their respective maps and metrics, as in \cite{nuneslarge}.

Positive-IC admission has its own estimation error. Under the same fixed-pool and iid-date assumptions, if $\widehat{\bar p}_i$ estimates the population mean IC $b_i$, then
\[
 \Prob\!\left(\max_i|\widehat{\bar p}_i-b_i|>\epsilon\right)
 \le2J e^{-T\epsilon^2/2}.
\]
Consequently, a uniform error allowance $\epsilon_{\rm IC}=\sqrt{2\log(2J/\alpha_{\rm IC})/T}$ makes the rule $\widehat{\bar p}_i\ge\beta+\epsilon_{\rm IC}$ certify $b_i\ge\beta$ with confidence $1-\alpha_{\rm IC}$. When both admission and pairwise separation are certified, allocate the total error budget between their two bounds. Positivity observed in the selection sample is not automatically positivity on future dates.

At a fixed confidence level, $\log J=o(T)$ makes this conservative uniform correlation bound vanish. If $J$ grows exponentially in $T$, the bound need not shrink. The fact that exponentially many histories \emph{can exist} does not ensure uniformly accurate estimation of all their relationships from $T$ dates. These are sufficient bounds, not impossibility results for every structured model.

For a shrinking angle, $1-\cos\theta\sim\theta^2/2$. To make this worst-case error negligible relative to that angular resolution, a sufficient joint condition is
\begin{equation}\label{eq:jointresolution}
 \frac{\log(J/\alpha)}{T\theta^4}\longrightarrow0.
\end{equation}
Variance-sensitive inequalities or additional structure may improve the requirement. Temporal dependence requires its own concentration assumptions; replacing $T$ by an informal effective sample size is not a proof. Small sample eigengaps and PnL covariance estimates likewise need their own calibration.

\section{How to investigate an actual selected library}

The empirical question is not whether the sphere has enough room. It is whether the admission rule spreads residual variation while preserving useful mean structure. A study should record the candidate pool, its orientation rule, the average-IC admission threshold and evaluation window, the order in which designs arrive, the pairwise correlation metric and threshold, and whether rejected designs may later be reconsidered. These choices define different populations even when the final threshold is identical.

Use separate displays for the number of admitted signals, the minimum observed separation, and an estimated covering radius over a declared candidate domain. A lack of remaining candidates can establish exhaustion of that pool; it does not establish saturation of all unit histories or a maximum packing. Randomizing candidate order measures the order sensitivity of greedy admission.

For each retained library, report its minimum and average mean IC and whether the margin persists as the library grows. Report also $\rho=\norm\mu^2$, $\eta=\norm\Sigma_{\op}$, the residual effective rank, the actual leading eigengap, and the EWS--PC angle. For the uniform margin benchmark, also compare the observed longitudinal and transverse blocks with the three levels in \cref{eq:marginlevels}; disagreement tests the benchmark assumptions. When feasible, compare the same quantities at individual dates. A history-level alignment can coexist with datewise disagreement, exactly as \cref{thm:alignedcode} demonstrates.

Next separate the signal Gram into its IC and transverse terms in \eqref{eq:payoff}. Compare uncentered IC similarity, temporally centered IC covariance, dispersion-weighted PnL covariance, and variance-normalized PnL correlation. A diversity rule on one of these matrices should not be described as a diversity rule on all of them. Evaluate predictive and spectral claims on held-out dates and include uncertainty appropriate to the time dependence and selection process.

The source bundle, supplied as ancillary files with the arXiv version, reproduces packing counts, moment formulas, finite codes, and margin spectra. It includes exact quadrature, finite-history gap checks, tilted convolution with grid refinement, checks of the fixed-margin correction, simulations of the central-limit regime, and perturbations of datewise unit-signal libraries. These calculations check the moment formulas and illustrate the proved results.

\section{What is settled, and what needs further research}

The phrase ``the entire sphere in the limit'' needs a selection rule. Saturation gives coverage as separation shrinks; near-maximum packing on a fixed rectifiable domain gives uniform volume. On the unrestricted sphere that limit has zero mean and isotropic second moment. Positive-IC admission changes the domain and preserves a positive mean, but the uniform positive half still has no unique PC1. Under uniform product-history volume, every positive margin makes the target uniquely leading at every finite history length. At a fixed margin, the mean norm tends to the margin and the gap to its square. A margin of order $T^{-1/2}$ retains a positive fraction of candidates, but its mean norm and absolute gap vanish.

The growing-history constructions establish that minimum separation permits both mean--PC agreement and exact disagreement in exponentially large libraries. A finite residual-spectrum condition supplies a positive alignment criterion. The target frame then identifies which surviving components affect realized payoff, without treating a coordinate choice as an independence assumption.

Three questions remain open. First, extend the established homogeneous RSA framework to IC-screened product domains, nonuniform or adaptive candidate arrivals, and specified stopping rules. Second, obtain packing-to-moment bounds that remain useful as history dimension grows, with certified occupancy deficits. Third, combine these with spectral inference under time dependence and vanishing means or eigengaps.

\phantomsection\addcontentsline{toc}{section}{Acknowledgments and AI-assistance disclosure}
\section*{Acknowledgments and AI-assistance disclosure}
The author developed this note through interactive work with ChatGPT (OpenAI), which assisted with mathematical derivations, literature retrieval, exposition, and numerical checks. Feedback from Claude (Anthropic), across multiple review passes, and from Gemini (Google) informed revisions. The note distinguishes established external results from derived applications and open questions. It has not undergone independent human peer review; the author is responsible for the final claims and citations.

\medskip
\appendix
\section{Strict target dominance at every finite history length}\label{app:finitegap}
Let $f_r$ be the density of one coordinate of a uniform point on $\Sphere^{r-1}$, extended by zero outside $[-1,1]$:
\[
 f_r(p)=\frac{\Gamma(r/2)}{\sqrt\pi\,\Gamma((r-1)/2)}(1-p^2)^{(r-3)/2},\qquad |p|<1.
\]
Write $b=T\beta$, $P=\Prob(\sum_t p_t\ge b)$, and define convolution densities
\[
 d_1=f_{q+2}*f_q^{*(T-1)},\qquad
 d_2=f_{q+2}*f_{q+2}*f_q^{*(T-2)}\quad(T\ge2).
\]
A zeroth convolution power is the unit point mass at zero. The exact gap identities are
\begin{equation}\label{eq:exactfinitegaps}
 \lambda_0-\lambda_\perp=\frac{\beta d_1(b)}{qP},\qquad
 \lambda_0-\lambda_L=-\frac{d_2'(b)}{q^2P}\quad(T\ge2).
\end{equation}
The first is strictly positive because $d_1(b)>0$ for $0<b<T$. We prove the second is strictly positive as well.

\subsection{Deriving the two gaps}
The elementary density identities
\[
 (1-p^2)f_q(p)=\frac{q-1}{q}f_{q+2}(p),\qquad
 pf_q(p)=-\frac1q f_{q+2}'(p)
\]
give the integration-by-parts formula
$\E[p_1 \phi(s)]=\E[(1-p_1^2)\phi'(s)]/(q-1)$, where $s=\sum_t p_t$. Apply it to $\phi(s)=s\one_{\{s\ge b\}}$, by smooth approximation. The interior term is $(1-A_{\beta,T})P/(q-1)$, the boundary term is $b d_1(b)/q$, and the left side is $T u_{\beta,T}P$. Rearrangement gives the transverse gap.

For distinct dates, the signed density weighted by $p_1p_2$ is $q^{-2}d_2''$. Integrating over $[b,\infty)$ gives
$\E[p_1p_2\one_{\{s\ge b\}}]=-d_2'(b)/q^2$.
The block calculation gives $\lambda_0-\lambda_L=\E[p_1p_2\mid s\ge b]$, proving the other identity. These operations are valid also at $q=2$: $f_{q+2}$ vanishes at the endpoints, its first derivative is integrable, and $d_2'$ is continuous and absolutely continuous. No finite density of $s$ at zero is required.

\subsection{Why the convolution derivative is negative}
We use an elementary one-dimensional fact about even densities. Represent an even nonincreasing density $g$, supported on $[-a,a]$, as a mixture of interval indicators,
$g(y)=\int\one_{\{|y|<r\}}\,\omega(dr)$. Then
\[
 (f*g)'(x)=\int [f(x+r)-f(x-r)]\,\omega(dr).
\]
If $f$ and $g$ strictly decrease through the interiors of their positive supports, the integrand is nonpositive and is negative on a set of positive $\omega$-measure for every $x$ inside the positive support of the convolution. The derivative is therefore negative there. The same conclusion holds when $g$ is uniform on $[-1,1]$ and the support radius of the strictly decreasing $f$ is at least one. The formula follows by differentiating the integral of $f$ over $[x-r,x+r]$; all convolutions used below have the stated differentiability.

For $q\ge3$, $f_{q+2}$ strictly decreases on $(0,1)$, and $f_q$ is either strictly decreasing there or uniform ($q=3$). Starting with $f_{q+2}*f_{q+2}$ and convolving with the remaining $f_q$ factors proves $d_2'(b)<0$ for $0<b<T$.

For $q=2$, $f_2$ itself is not unimodal, so a separate argument is needed. Both $g_{22}=f_2*f_2$ and $g_{42}=f_4*f_2$ are even and strictly decreasing on $(0,2)$. For the first, setting $v=x/2\in(0,1)$ in its convolution integral gives
\[
 g_{22}(x)=\frac{2}{\pi^2(1+v)}\int_0^1
 \frac{du}{\sqrt{1-u^2}\sqrt{1-[(1-v)/(1+v)]^2u^2}}.
\]
The prefactor and integrand decrease as $v$ increases. For the second, $f_4$ is the first-coordinate density of a uniform unit disk and $f_2$ that of a uniform unit circle. Their planar sum has radial density
$f_{\rm rad}(r)=\pi^{-2}\arccos(r/2)$ for $0<r<2$, and zero for $r\ge2$. This follows by measuring the circle arc lying inside a unit disk centered at a point of radius $r$. The density is radially strictly decreasing; its projection $g_{42}(x)=\int_{\R}f_{\rm rad}(\sqrt{x^2+y^2})\,dy$ is consequently strictly decreasing for $0<x<2$.

For even $T$, write $d_2=f_4*f_4*g_{22}^{*(T-2)/2}$; for odd $T\ge3$, write $d_2=f_4*g_{42}*g_{22}^{*(T-3)/2}$. The interval-mixture argument again gives $d_2'(b)<0$ throughout $(0,T)$. This completes \cref{thm:margin} for every $q\ge2$. It proves strict ordering, not monotonic growth of the eigengap as $\beta$ changes.

\section{The conditional margin limit}\label{app:tilting}
Here $q$ and $0<\beta<1$ are fixed. Let $F$ be the uniform-sphere coordinate law, $\Lambda(\kappa)=\log\E_F e^{\kappa p}$, and $I(b)=\kappa_b b-\Lambda(\kappa_b)$, where $\Lambda'(\kappa_b)=b$. Strict convexity gives $I'(b)=\kappa_b>0$ for $0<b<1$.

We first establish the elementary tail limit
\begin{equation}\label{eq:tailrate}
 \frac1T\log\Prob_F(c\ge b)\longrightarrow-I(b),\qquad 0<b<1.
\end{equation}
Exponential Markov gives the upper bound. For a lower bound, choose $b'>b$ and $\epsilon>0$ with $b'-\epsilon>b$. Under the product tilted law $F_{\kappa_{b'}}^{\otimes T}$, the law of large numbers makes $|c-b'|\le\epsilon$ have probability tending to one. On this event the inverse likelihood ratio is at least
$\exp\{-T[\kappa_{b'}(b'+\epsilon)-\Lambda(\kappa_{b'})]\}$.
Let $\epsilon\downarrow0$ and then $b'\downarrow b$ to obtain the matching lower bound.

For every $\epsilon>0$ with $\beta+\epsilon<1$, \eqref{eq:tailrate} and strict increase of $I$ imply
\[
 \Prob(c\ge\beta+\epsilon\mid c\ge\beta)
 =\exp\{-T[I(\beta+\epsilon)-I(\beta)+o(1)]\}\longrightarrow0.
\]
Thus $c\to\beta$ conditionally. To identify one-date moments without assuming conditional independence, let $P_{\beta,T}$ be the conditional joint law of $(p_1,\ldots,p_T)$ and set $\kappa=\kappa_\beta$. Relative entropy gives the exact identity
\[
 D(P_{\beta,T}\Vert F_\kappa^{\otimes T})
 =-\log\Prob(c\ge\beta)-\kappa Tm_{\beta,T}+T\Lambda(\kappa)=o(T).
\]
By exchangeability all one-coordinate marginals are the same law $F_{\beta,T}^{(1)}$. The entropy decomposition against a product law gives
\[
 D(P_{\beta,T}\Vert F_\kappa^{\otimes T})
 =D\bigl(P_{\beta,T}\Vert (F_{\beta,T}^{(1)})^{\otimes T}\bigr)
 +T D(F_{\beta,T}^{(1)}\Vert F_\kappa).
\]
Nonnegativity therefore forces $D(F_{\beta,T}^{(1)}\Vert F_\kappa)\to0$. Pinsker's inequality and $|p|\le1$ give $A_{\beta,T}\to\E_\kappa p^2$. This proves the moment limits used in \cref{cor:marginlimit}. The method is the usual exponential-tilting and conditional-limit argument; the calculation here is supplied to make its application to the product-sphere screen explicit.

\subsection{Finite-history corrections and small margins}\label{app:refinements}
For fixed $q\ge2$ and $0<\beta<1$, put $\kappa=\kappa(q,\beta)$ and $v_\kappa=\Var_\kappa p$. The precise admission estimate is
\begin{equation}\label{eq:admissionprefactor}
 \Prob(c\ge\beta)\sim\frac{e^{-T I(\beta)}}{\kappa\sqrt{2\pi T v_\kappa}}.
\end{equation}
This is the nonlattice Bahadur--Ranga Rao prefactor \cite{bahadurrao}; see also \cite[Section 2]{ght}. The following local-limit argument also proves the moment correction \eqref{eq:finitecorrection}.

Let $l_T$ be the density of $Y_T=\sum_t p_t-T\beta$ under the tilted product law. The uniform density local central limit theorem gives
\[
 \sup_y\left|\sqrt{Tv_\kappa}\,l_T(y)-\varphi\!\left(\frac{y}{\sqrt{Tv_\kappa}}\right)\right|\longrightarrow0.
\]
Its hypotheses hold here: the tilted coordinate has bounded support, positive variance, and an absolutely continuous nonlattice law. Its characteristic function is $O(|u|^{-1/2})$ or faster, as follows by splitting the endpoint integrals and integrating by parts away from them; hence a convolution power has an integrable characteristic function. Fourier inversion, the quadratic expansion near zero, and a bound strictly below one away from zero give the displayed uniform limit. In particular, $\sqrt T\,l_T$ is uniformly bounded for all sufficiently large $T$.

Changing measure and using this bound, dominated convergence gives, for $j=0,1,2$,
\[
 e^{T I(\beta)}\E_F[Y_T^j\one_{\{Y_T\ge0\}}]
 =\int_0^\infty y^j e^{-\kappa y}l_T(y)\,dy
 \sim\frac{j!}{\kappa^{j+1}\sqrt{2\pi T v_\kappa}}.
\]
Here $Y_T$ denotes the same sum functional under the original law on the left. The $j=0$ case proves \eqref{eq:admissionprefactor}; dividing the other two cases by it proves conditional moment convergence to an exponential law of rate $\kappa$. The distributional limit follows in the same way using bounded test functions. Since $c=\beta+Y_T/T$, these moments give \eqref{eq:finitecorrection}. No uniformity as $\beta\downarrow0$ is asserted.

For fixed $q$ and $\beta\downarrow0$, the limiting coefficients in \cref{cor:marginlimit} satisfy
\begin{equation}\label{eq:smallbeta}
\begin{aligned}
 \kappa&=q\beta+\frac{q^2}{q+2}\beta^3+O(\beta^5),\\
 v_\kappa&=\frac1q-\frac{3\beta^2}{q+2}+O(\beta^4),\qquad
 \frac{\beta}{\kappa}=\frac1q-\frac{\beta^2}{q+2}+O(\beta^4),\\
 I(\beta)&=\frac{q\beta^2}{2}+\frac{q^2\beta^4}{4(q+2)}+O(\beta^6).
\end{aligned}
\end{equation}
These follow from $\E p^2=1/q$, $\E p^4=3/[q(q+2)]$, and
$\Lambda(\kappa)=\kappa^2/(2q)-\kappa^4/[4q^2(q+2)]+O(\kappa^6)$, followed by inversion of $\Lambda'(\kappa)=\beta$. Also $v_\kappa=d\beta/d\kappa$ and $I'(\beta)=\kappa$. Both residual coefficients approach $1/q$, consistent with the positive-half benchmark. Their difference begins at order $\beta^2$.

These expansions concern the fixed-margin limiting coefficients. Their residual difference is $2\beta^2/(q+2)+O(\beta^4)$, so for a sufficiently small margin held fixed before $T\to\infty$, the transverse level is eventually the larger residual. Substituting a shrinking $\beta_T$ into this statement would require additional error control.

For small fixed $\beta$, $2\beta/\kappa=2/q+O(\beta^2)$, explaining the useful approximation $\lambda_0\approx\beta^2+2/(qT)$ in a long history. In the normal reference, integration by parts gives
$h(z)=z+z^{-1}-2z^{-3}+O(z^{-5})$, hence
$1+zh(z)=z^2+2-2z^{-2}+O(z^{-4})$.
The same constant $2$ appears in the large-$z$ expansion of the central-limit expression. This agreement checks the two formulas; it is not a theorem uniform over an arbitrary joint path $z\to\infty$, $\beta\to0$.

The rarity calculation needs the same care. The leading small-$\beta$ exponential term is $e^{-Tq\beta^2/2}$, but it is not a finite-sample tail probability. In the normal reference, with $z=\sqrt{qT}\beta$,
\[
 \overline\Phi(z)\sim\frac{e^{-z^2/2}}{z\sqrt{2\pi}}
 \quad(z\to\infty).
\]
For $q=19$, $\beta=0.1$, and $T=100$, refined numerical convolution gives admission probability $6.29\times10^{-6}$. The normal reference gives $6.54\times10^{-6}$; the exponential term alone gives $7.49\times10^{-5}$. Omitting the prefactor overstates admission by about an order of magnitude. These numerical comparisons do not assert a uniform normal-tail approximation in arbitrary joint limits.

\section{Proof at the central-limit scale}\label{app:localscreen}
Set $Y_T=\sqrt{q/T}\sum_{t=1}^T p_t$. For fixed $q$, the iid uniform-sphere coordinates have mean zero, variance $1/q$, and fourth moment $3/[q(q+2)]$. Therefore
\[
 Y_T\Rightarrow Z\sim N(0,1),\qquad
 \E Y_T^4=3-\frac{6}{T(q+2)}\le3.
\]
For fixed $z\ge0$, the event $Y_T\ge z$ has probability tending to $\overline\Phi(z)>0$. The fourth-moment bound gives uniform integrability of $Y_T^2$, so truncation and weak convergence yield
\[
 \E[Y_T\mid Y_T\ge z]\to h(z),\qquad
 \E[Y_T^2\mid Y_T\ge z]\to1+zh(z).
\]
The normal identities follow by integrating $y\varphi(y)$ and $y^2\varphi(y)$ over $[z,\infty)$.

For the one-date moment, write $Y_T=W_T+\sqrt{q/T}\,p_1$, where $W_T$ is independent of $p_1$ and $W_T\Rightarrow N(0,1)$. Since $|p_1|\le1$, the conditional admission probability given $p_1$ is bounded between the tails of $W_T$ at $z\pm\sqrt{q/T}$. Both bounds tend to $\overline\Phi(z)$. It follows by bounded convergence that
\[
 \E[p_1^2\mid Y_T\ge z]\longrightarrow\E p_1^2=1/q.
\]
Now $qT u_{\beta_T,T}\to1+zh(z)$, $\sqrt{qT}\,m_{\beta_T,T}\to h(z)$, and $A_{\beta_T,T}\to1/q$. Substitution into \eqref{eq:marginlevels} proves \cref{prop:localscreen}. In particular, for $z>0$ the scaled gap has a strictly positive limit, even though the unscaled gap tends to zero.

\phantomsection
\end{document}